\documentclass[journal,twoside,web]{ieeecolor}
\usepackage{generic}
\usepackage{cite}
\usepackage{hyperref}
\hypersetup{hidelinks=true}
\usepackage{textcomp}
\usepackage{graphicx}

\usepackage{cite}
\usepackage[inkscapelatex=false]{svg}
\usepackage{cite}
\usepackage{amsmath,amssymb,amsfonts}
\usepackage{graphicx}
\usepackage{textcomp}
\usepackage{setspace}

\usepackage{svg}

\usepackage{bm}
\usepackage{tabularx}
\usepackage{multirow}
\usepackage{booktabs}

\usepackage{amsthm}

\newcommand{\Scal}{\mathcal{S}}

\theoremstyle{plain}
\newtheorem{theorem}{Theorem}
\newtheorem{corollary}{Corollary}
\newtheorem{lemma}{Lemma}
\newtheorem{problem}{Problem}
\newtheorem{definition}{Definition}
\newtheorem{prop}{Proposition}

\theoremstyle{definition}

\newtheorem{remark}{Remark}

\makeatletter
\let\NAT@parse\undefined
\makeatother
\usepackage{hyperref}
\hypersetup{
colorlinks, 
    linkcolor={red!50!black},
    citecolor={blue!50!black},
    urlcolor={blue!80!black}
}
\usepackage{cleveref}

\crefformat{problem}{Problem~#2#1#3}
\crefformat{prop}{Proposition~#2#1#3}

\def\BibTeX{{\rm B\kern-.05em{\sc i\kern-.025em b}\kern-.08em
    T\kern-.1667em\lower.7ex\hbox{E}\kern-.125emX}}
\begin{document}
\title{Minimal Construction of Graphs with Maximum Robustness}
\author{Haejoon Lee, \IEEEmembership{Student Member, IEEE}, and Dimitra Panagou, \IEEEmembership{Senior Member, IEEE}
\thanks{*This work was supported by the Air Force Office of Scientific Research (AFOSR) under Award No. FA9550-23-1-0163.}
\thanks{Haejoon Lee is with the Department of Robotics, University of Michigan, Ann Arbor, MI, USA. Dimitra Panagou is with the Department of Robotics, and with the Department of Aerospace Engineering, University of Michigan, Ann Arbor, MI, USA. (email: haejoonl@umich.edu; dpanagou@umich.edu).}}
\maketitle

\begin{abstract}
The notions of $r$-robustness and $(r,s)$-robustness of a network have been earlier introduced in the literature to achieve resilient consensus in the presence of misbehaving agents. However, while higher robustness levels enable networks to tolerate a higher number of misbehaving agents, they also require dense communication structures, which are not always desirable for systems with limited communication ranges, energy, and resources. Therefore, this paper studies the fundamental structures behind $r$-robustness and $(r,s)$- robustness properties in two ways. (a) We first establish tight necessary conditions on the number of edges that an undirected graph with an arbitrary number of nodes must have to achieve maximum $r$- and $(r,s)$-robustness. (b) We then use these conditions to construct two classes of undirected graphs, referred to as $\gamma$- and $(\gamma,\gamma)$-Minimal Edge Robust Graphs (MERGs), that provably achieve maximum robustness with minimal numbers of edges. We demonstrate the effectiveness of our method via comparison against existing robust graph structures and a set of simulations.
\end{abstract}

\begin{IEEEkeywords}
Network control systems, distributed consensus and control, resilient consensus
\end{IEEEkeywords}

\section{Introduction}
\label{sec:introduction}
\IEEEPARstart{D}{istributed} multi-agent systems have many practical uses, ranging from information gathering \cite{viseras2018}, target tracking \cite{zhang2021}, to collaborative decision making \cite{yanhao2021}. At the core of many distributed algorithms lies consensus, where multiple agents reach agreement on common state values. However, it is known that performance of consensus deteriorates significantly when one or more misbehaving agents share wrong, or even adversarial, information with the system. As a result, resilient consensus in the presence of misbehaving agents has been studied extensively \cite{LeBlanc13,yan2022,zhang2012,aydın2024,zhang2015,usevitch2020FiniteTime,abbas_vector_resilient_2022, ICRA2025}.
In \cite{LeBlanc13,zhang2012,zhang2015}, an algorithm called \textit{Weighted Mean Subsequence Reduced} (W-MSR) algorithm was introduced to allow the non-misbehaving agents (often called \textit{normal}) to reach consensus despite the presence of misbehaving agents. In contrast to reputation- and detection-based approaches~\cite{hadjicostis_trustworthy_2025,yuan_resilient_2025,yemini2022} that require high computational overhead, the W-MSR algorithm instead ensures resilience through a low-complexity filtering strategy. 

The notions of $r$- and $(r,s)$-robustness have been introduced in~\cite{LeBlanc13, zhang2012} to characterize the resilience guarantees provided by the W-MSR algorithm. In particular, they quantify how many misbehaving agents the network tolerates while ensuring normal agents in a network to reach consensus within the convex hull of their initial values through the W-MSR algorithm~\cite{LeBlanc13}. $(2F+1)$-robustness provides a sufficient condition for normal agents to achieve consensus in the presence of $F$-local Byzantine agents. In contrast, $(F+1,F+1)$-robustness serves as a necessary and sufficient condition for tolerating $F$-total malicious agents~\cite{LeBlanc13}. Due to their distinct resilience guarantees, both $r$- and $(r,s)$-robustness have attracted significant attention and have been recently extended to ensure resilience in distributed estimation~\cite{liewi2021byzantine},  optimization~\cite{sundaram2019distributed_opt,yuan2025resilient}, and learning~\cite{xie2021towards, ye2024resilient}. Their effectiveness has been demonstrated across a broad range of networked systems, including cooperative robotics and smart grids~\cite{CDC2025,li2025distributed,  cavorsi2022, yuan2025resilient}.

Consequently, maximizing the robustness of a graph is highly desirable. However, achieving higher robustness by definition requires a higher number of edges, often referred to as edge counts. This results in dense and complex graph structures, which pose challenges for certain systems due to practical constraints such as limited communication range, bandwidth, and energy. This fundamental tradeoff between robustness and edge counts motivates us to investigate the following question: \textit{How can we systematically construct communication graphs that achieve maximum robustness with minimal edge counts?}

Despite its clear practical importance, identifying the fundamental structures of robust graphs is particularly challenging due to the combinatorial nature of the problem. In fact, determining $r$-robustness is a coNP-complete problem \cite{zhang2015}. Although algorithms for computing exact robustness exist \cite{usevitch2020, leblanc2013alg}, they scale poorly with the number of nodes. This computational bottleneck makes brute-force search approaches impractical. Although there have been works to speed up computations with sampling-based
\cite{sample_based_r_computation}, heuristics \cite{heuristic_robustness_computation}, or machine learning-based approaches \cite{learning_robustness_computation}, they either provide approximations without guarantees or require specific assumptions about graph topologies. In addition, other works under-bound a given graph's robustness with other graph properties that are generally simpler and faster to compute~\cite{LeBlanc13,cavorsi2022,saulnier2017,CDC2025, wang_resilient_2022}. However, these bounds are in general too loose for many graph structures.

To circumvent the combinatorial issue, some researches have focused on the construction and analysis of specific graph classes with predetermined $r$- and $(r,s)$-robustness properties. A complete graph with $n$ nodes is $\lceil\frac n 2\rceil$-robust, indicating that $\lceil\frac n 2\rceil$ is the maximum level of $r$-robustness any graph of $n$ nodes achieves~\cite{LeBlanc13}. Authors of \cite{usevitch2017circulant} analyzed robustness of a class of circulant digraphs. A preferential-attachment method is investigated in \cite{zhang2012,LeBlanc13, zhang2015} for scalable robust graph constructions. Furthermore, robustness levels of various graph structures, such as geometric and random graphs \cite{Guerrero20, Guerrero19,zhang2015,zhao2017,Shahrivar2017, sood2025balancing} have been explored.

Still, understanding the fundamental structural properties of general $r$- and $(r,s)$-robust graphs remains an open problem. While the authors in \cite{saldana2016} have established the lower bounds on edge counts and construction mechanisms for undirected $(2,2)$-robust graphs, these results do not generalize to other robustness levels. The authors in \cite{Guerrero17} study the construction of a class of undirected $r$-robust graphs with the minimum number of nodes. While this is equivalent to constructing graphs with maximum $r$-robustness for a given number of nodes, they do not provide any guarantees on the minimality of the edge set, and their construction method is valid only for specific values of $n$. The concept of minimal edge sets for $(r,s)$-robust graphs has been introduced in \cite{wehbe21}, though specific scalable construction mechanisms are not discussed.

Instead of investigating the fundamental structures of robust graphs, several studies have sought to relax the connectivity requirements associated with robustness. Some have explored time-varying network topologies to reduce connectivity requirements \cite{saldana2017,wen2023,yu2022}. In \cite{abbas2018}, authors introduced trusted agents, who are always deemed fault-free, to achieve robustness with fewer edges. Furthermore, robustness under multi-hop communication schemes has been investigated in \cite{yuan2021}. More recently, authors in~\cite{gao2025activesecureneighborselection} have proposed relaxing the connectivity requirements by employing detection-based isolation strategies. While these reduce the connectivity requirements associated with $r$- and $(r,s)$-robustness, they only do so by lowering the required robustness level and thus still depend on robust graph structures. This further underscores the need to investigate the fundamental structural properties of robustness.

\subsubsection*{Contributions} 
Our work directly investigates the fundamental structures of $r$- and $(r,s)$-robust graphs in two primary ways. First, we derive tight lower bounds on edge counts required for graphs with any number of nodes to achieve maximum robustness. To the best of our knowledge, these bounds are the first to explicitly depend on the level of robustness. Then, for an arbitrary number of nodes, we present two classes of graphs that provably achieve maximum robustness with the minimum number of edges. We refer to these graphs as $\gamma$- and $(\gamma,\gamma)$-Minimal Edge Robust Graphs (MERGs)\footnote{In~\cite{CDC2024}, we referred to graphs that achieve maximum $r$-robustness with minimal sets of edges as $(2r-1,r)$- and $(2r,r)$-robust graphs. To avoid confusion with the standard definition of $(r,s)$-robustness, we adopt a new naming convention in this paper.}. We build upon our previous work \cite{CDC2024}, in which we only studied $r$-robust graphs with minimal sets of edges, by also considering $(r,s)$-robust graphs. Specifically, we make the following contributions:
\begin{enumerate}
    \item We establish tight lower bounds on the number of edges that an undirected graph with any number of nodes must have to achieve maximum $r$- or $(r,s)$-robustness.
    \item Based on the lower bounds, we construct two classes of undirected graphs with an arbitrary number of nodes, referred to as $\gamma$- and $(\gamma,\gamma)$-MERGs that are proven to achieve maximum $r$- and $(r,s)$-robustness with minimal edge sets, respectively. Moreover, we prove that the construction mechanisms for $(\gamma,\gamma)$-MERGs are both necessary and sufficient, thereby capturing the entire space of graphs that achieve maximum $(r,s)$-robustness with the fewest edges.
    \item We evaluate our methods through comparisons with other classes of robust graphs and a series of simulations.
\end{enumerate}

The remainder of this paper is organized as follows. Section~\ref{sec:prelim} presents preliminaries and the problem statement. Section~\ref{sec:lower_bounds1} establishes necessary conditions on edge counts for graphs to achieve maximum $r$-robustness, whereas Section~\ref{sec:construction1} provides systematic constructions of $\gamma$-MERGs. Section~\ref{sec:lower_bounds2} determines the necessary conditions on edge counts for graphs to achieve maximum $(r,s)$-robustness, and Section~\ref{sec:construction2} presents constructions of $(\gamma,\gamma)$-MERGs. Section~\ref{sec:comparison} compares the proposed constructions with other existing graph constructions. Section~\ref{sec:sim} presents simulation results, and Section~\ref{sec:conc} concludes the paper and discusses directions for future research.

\section{Background and Problem Statement}
\label{sec:prelim}
\subsection{Notation}
We denote a simple, undirected time-invariant graph as $\mathcal G = (\mathcal {V},\mathcal{E})$ where $\mathcal {V}$ and $\mathcal{E}$ are the vertex set and the undirected edge set of graph respectively. We denote the number of nodes $n$. An undirected edge $(i,j)\in \mathcal{E}$ indicates that information is exchanged between nodes $i$ and $j$. The neighbor set of agent $i\in \mathcal V$ is denoted as $\mathcal N_i =\{j\in \mathcal V \;|\; (i,j) \in \mathcal{E}\}$. The state of agent $i$ at time $t$ is denoted as $x_i[t]$. We denote the cardinality of a set $\mathcal S$ as $|\mathcal S|$. We denote the set of non-negative and positive integers as $\mathbb Z_{\geq 0}$ and $\mathbb Z_{>0}$. Lastly, we denote by $\lceil \cdot \rceil$ and $\lfloor \cdot \rfloor$ the ceiling and floor operators, respectively.

\subsection{Weighted Mean Subsequence Reduced Algorithm}
In this section, we review the relevant fundamental concepts of $r$- and $(r,s)$-robustness of a graph, as well as the W-MSR algorithm, which are defined in \cite{LeBlanc13}. 

Let $\mathcal G = (\mathcal V, \mathcal E)$ be a graph. Then, an agent $i \in \mathcal V$ shares its state $x_i[t]$ at time $t$ with all neighbors $j \in \mathcal N_i$, and each agent $i$ updates its state according to the nominal update rule:
\begin{equation}
\label{eq:linear}
    x_i[t+1]=w_{ii}[t]x_i[t] + \sum_{j\in \mathcal N_i} w_{ij}[t]x_j[t],
\end{equation}
where $w_{ij}[t]$ is the weight assigned to agent $j$'s value by agent $i$, and where the following conditions are assumed for $w_{ij}[t]$ for $\forall i \in \mathcal V$ and $t \in \mathbb Z_{\geq0}$:
\begin{itemize}
    \item $w_{ij}[t]=0$ $\forall j\notin \mathcal N_i\cup \{i\}$,
    \item $w_{ij}[t]\geq \alpha$, $0\leq \alpha <1$ $\forall j\in \mathcal N_i\cup \{i\}$,
    \item $\sum_{j=1}^n w_{ij}[t]=1$
\end{itemize}
Through the protocol given by \eqref{eq:linear}, agents are guaranteed to reach consensus on their states as long as the graph is connected~\cite{wei2007, sundaram2008}. However, consensus may not be reached in the presence of misbehaving agents~\cite{LeBlanc13}. In our paper, we assume that the identities of misbehaving agents are unknown, and we focus on two types of misbehaving agents:
\begin{definition}[\textbf{Malicious agent}]
    \label{def:malicious}
    An agent $i\in \mathcal V$ is called \textbf{malicious} if it does not follow the nominal update protocol \eqref{eq:linear} at some time step $t$.
\end{definition}
\begin{definition}[\textbf{Byzantine agent}]
    \label{def:Byzantine}
    An agent $i\in \mathcal V$ is called \textbf {Byzantine} if it does not follow the nominal update protocol \eqref{eq:linear}, or if it
does not send the same value to all of its neighbors at some time step $t$.
\end{definition}
A malicious agent may deviate from the nominal update law but still broadcasts the \textit{same} value to all of its neighbors. In contrast, a Byzantine agent may both deviate from the protocol and send different values to different neighbors. Hence, a malicious agent is as a special, weaker case of a Byzantine agent. Agents that are not misbehaving are called normal agents.
There are numerous models to describe the number of misbehaving agents in a network \cite{LeBlanc13}. We focus on two such models below:
\begin{definition}[$\mathbf F$\textbf{-total}]
    \label{def:ftotal}
    A set $\mathcal S \subset \mathcal V$ is \textbf{$\mathbf F$-total} if it contains at most $F$ nodes in the graph (i.e. $|\mathcal S|\leq F)$.
\end{definition}

\begin{definition}[$\mathbf F$\textbf{-local}]
    \label{def:flocal}
    A set $\mathcal S \subset \mathcal V$ is \textbf{$\mathbf F$-local} if all other nodes have at most $F$ nodes of $\mathcal S$ as their neighbors (i.e. $|\mathcal N_i\cap \mathcal S|\leq F$, $\forall i \in \mathcal V \setminus \mathcal S$).
\end{definition}

In response to various threat and scope models, algorithms on resilient consensus \cite{LeBlanc13, saldana2017,dibaji2017,yan2022,hadjicostis_trustworthy_2025,yuan_resilient_2025} have become very popular. In particular, the W-MSR (Weighted Mean Subsequence Reduced) algorithm \cite{LeBlanc13} with the parameter $F$ guarantees normal agents to achieve asymptotic consensus with $F$-total or $F$-local misbehaving agents under certain assumed topological properties of the communication graph. 

\begin{definition}[$\mathbf r$\textbf{-reachable}]
    \label{def:r_reachability}
    Let $\mathcal G = (\mathcal V,\mathcal{E})$ be a graph and $\mathcal S$ be a nonempty subset of $\mathcal V$. The subset $\mathcal S$ is $\mathbf r$\textbf{-reachable} if $\exists i\in \mathcal S$ such that $|\mathcal N_i \backslash \mathcal S|\geq r$, where $r\in \mathbb Z_{>0}$.
\end{definition}

\begin{definition}[$\mathbf r$\textbf{-robust}]
    A graph $\mathcal G = (\mathcal V,\mathcal{E})$ is $\mathbf r$\textbf{-robust} if $\forall \mathcal S_1,\mathcal S_2 \subset \mathcal V$ where $\mathcal S_1\cap S_2 = \emptyset$ and $\mathcal S_1,\mathcal S_2\neq \emptyset$, at least one of them is $r$-reachable.
    \label{def:r_robust}
\end{definition}

\begin{definition}[$\mathbf {(r,s)}$\textbf{-reachable}]
    \label{def:rs_reachability}
  Given a graph $\mathcal G=(\mathcal V, \mathcal E)$ and
a nonempty subset $\mathcal S \subset \mathcal V$, $\mathcal S$ is $\mathbf{(r,s)}$\textbf{-reachable} if there are at least $s \in \mathbb Z_{\geq 0}$ nodes in $\mathcal S$ with at least $r \in \mathbb Z_{> 0}$ neighbors outside of $\mathcal S$; i.e., given $\mathcal X^r_{\mathcal S} = \{i \in \mathcal S : |\mathcal N_i \setminus \mathcal S| \geq r\}$, then $|\mathcal X^r_{\mathcal S}|\geq s$.

\end{definition}

\begin{definition}[$\mathbf {(r,s)}$\textbf{-robust}]
    A graph $\mathcal G = (\mathcal V,\mathcal{E})$ is $\mathbf {(r,s)}$\textbf{-robust} if for every
pair of nonempty, disjoint subsets $\mathcal S_1,\mathcal S_2 \subset \mathcal V$, at least one
of the following holds (let $\mathcal X^r_{\mathcal S_k} = \{i \in \mathcal S_k : |\mathcal N_i \setminus \mathcal S_k| \geq r\}$ for $k\in\{1,2\}$):
\begin{itemize}
    \item $|\mathcal X^r_{\mathcal S_1}|= |\mathcal S_1|$;
    \item $|\mathcal X^r_{\mathcal S_2}|= |\mathcal S_2|$;
    \item $|\mathcal X^r_{\mathcal S_1}|$ + $|\mathcal X^r_{\mathcal S_2}|\geq s$.
\end{itemize}
    \label{def:rs_robust}
\end{definition}

A set $\Scal$ is $r$-reachable if it contains at least one node with at least $r$ neighbors outside the set, and $(r,s)$-reachable if it contains at least $s$ such nodes. These notions help us define $r$- and $(r,s)$-robustness. A graph is $r$-robust if, for \textit{every} pair of disjoint subsets, at least one subset is $r$-reachable. In contrast, a graph is $(r,s)$-robust if, for \textit{every} pair of disjoint subsets, at least $s$ nodes have at least $r$ neighbors outside  their respective subsets. Hence, $(r,s)$-robustness generalizes $r$-robustness: any $(r,s)$-robust graph with $s\geq1$ is also $r$-robust. Conversely, if a graph is $r$-robust, it is at least $(r,1)$-robust.

While $(2F+1)$-robustness provides a sufficient condition for consensus under $F$-local Byzantine agents, $(F+1,F+1)$-robustness is both a necessary and sufficient condition to tolerate $F$-total malicious agents~\cite{LeBlanc13}. Note that these notions represent stronger forms of network resilience compared to more traditional graph-theoretic measures such as connectivity or minimum degree~\cite{LeBlanc13, pirani2023}. In particular, $r$- and $(r,s)$-robust networks are also at least $r$-connected and have a minimum degree of $r$~\cite{LeBlanc13, cavorsi2022}. Thus, these notions of robustness enable a network to achieve the same level of resilience (i.e., tolerance to the same number of misbehaving agents) using fewer edges (cf.~\cite[Fig. 3]{pirani2023}).

\subsection{Problem Statement}
In this paper, we aim to construct graphs with a given number of nodes that achieve maximum $r$- and $(r,s)$- robustness using the minimal number of edges. Consider a graph $\mathcal G=(\mathcal V, \mathcal E)$ where $|\mathcal V|=n$. Before formally defining our problems, we first define two classes of graphs below:

\begin{definition}
    [$\mathbf r$\textbf{-MERG}]
    \label{def:rmerg}
    A graph $\mathcal G$ is $\mathbf r$\textbf{-Minimal Edge Robust Graph}, or $\mathbf r$\textbf{-MERG} if $\mathcal G$ achieves $r$-robustness with the fewest edges possible.  
\end{definition}

\begin{definition}
    [$\mathbf {(r,s)}$\textbf{-MERG}]
    \label{def:rsmerg}
    A graph $\mathcal G$ is $\mathbf {(r,s)}$\textbf{-Minimal Edge Robust Graph}, or $\mathbf {(r,s)}$\textbf{-MERG} if $\mathcal G$ achieves $(r,s)$-robustness with the fewest edges possible.
\end{definition}

The $r$-MERGs and $(r,s)$-MERGs are general classes of graphs that achieve certain robustness levels with minimal sets of edges. Note that in~\cite{CDC2024} we have provided definitions of the same properties using different terms, which however may cause ambiguity and appear conflicting to the well-established notion of $(r,s)$-robust graphs. Hence in this paper we introduce the MERG definitions above to avoid any confusion. Since the maximum $r$-robustness $\mathcal G$ can achieve is $r=\lceil\frac {n}{2}\rceil $~\cite{LeBlanc13}, the problems addressed in this paper are:

\begin{problem}
    Given a desired number of nodes $n\in \mathbb Z_{>0}$,
    construct $\lceil \frac n 2\rceil$-MERGs.
    \label{prob:first}
\end{problem}
\begin{problem}
 Given a desired number of nodes $n\in \mathbb Z_{>0}$, construct $(\lceil \frac n 2\rceil,\lceil \frac n 2\rceil)$-MERGs.
    \label{prob:second}
\end{problem}

\begin{remark}
While the actual maximum $(r,s)$-robustness $\mathcal G$ achieves is $(\lceil \frac n 2 \rceil, n)$, this is only achieved by complete graphs~\cite[Lemma 4]{LeBlanc13}, which are rather trivial. Instead, we focus on the case where $r=s=\lceil \frac n 2 \rceil$, as $(\lceil \frac n 2 \rceil,\lceil \frac n 2 \rceil)$-robust network tolerates up to $\lceil \frac n 2 \rceil-1$ malicious agents, which is the maximum number of malicious agents the network with $n$ nodes tolerates through the W-MSR algorithm in practice.
\end{remark}

For ease of notation, we define $\gamma = \lceil \frac{n}{2} \rceil$ for the remainder of the paper. Systematically constructing $\gamma$- and $(\gamma,\gamma)$-MERGs allows us to design networks that achieve the highest possible resilience for many distributed algorithms to misbehaving agents with the minimal number of communication links. Such constructions allow the system across diverse application domains to enjoy the highest robustness with reduced communication and energy overhead, which is critical for resource-constrained networks such as sensor networks, robotic swarms, and other distributed control systems. Moreover, the proposed constructions provide graph topologies that also work for extensions of the W-MSR algorithm which address practical network constraints, such as time-varying topologies and limited-rate communication~\cite{wang2020resilient, saldana2017}, as well as multi-dimensional states~\cite{shang2020resilient}, without requiring additional edges.

Before proceeding with our analyses, we note that for both $\gamma$- and $(\gamma,\gamma)$-MERGs, we separate the analyses into two cases based on whether $n$ is even or odd. This separation is made for analytical convenience in our approach. For both cases, our analysis takes place in two steps: First, we establish the lower bounds on the edge counts for a graph with $n$ nodes to achieve maximum robustness ($\gamma$- and $(\gamma,\gamma)$-robustness). Second, we leverage these lower bounds to develop construction mechanisms that are proven to construct $\gamma$- and $(\gamma,\gamma)$-MERGs.

\section{Bounds on Edge Counts for $\gamma$-MERGs}
\label{sec:lower_bounds1}
This section addresses the first step toward constructing $\gamma$-MERGs by examining the minimum number of edges required for a graph $\mathcal G=(\mathcal V,\mathcal E)$ with $|\mathcal V|=n$ to achieve $\gamma=\lceil \frac n 2 \rceil$-robustness. This section considers two cases. In the first case, we find the lower bounds on edge counts for $\gamma$-robust graphs with odd $n$. Then, in the second case, we do the same for $\gamma$-robust graphs with even $n$. These are formally established in~\Cref{lem:min_addition} and~\Cref{lem:min_addition2}, respectively. These lemmas will be later used to construct $\gamma$-MERGs in Section~\ref{sec:construction1}.

Before analyzing the lower bounds on the edge counts for $\gamma$-robust graphs, we first introduce the concept of clique:
\begin{definition}[clique \cite{clique}]
    A \textbf{clique} $\mathcal C$ of a graph $\mathcal G = (\mathcal V, \mathcal{E})$ is a subset $\mathcal{C} \subseteq \mathcal{V}$ such that $\forall i, j \in \mathcal{C}$ with $i \ne j$, we have $(i, j) \in \mathcal{E}$.
\end{definition}
In general, $k$-clique refers to a clique consisting of $k$ nodes. The clique with the largest cardinality in $\mathcal G$ is called the maximum clique. We will use cliques to discover necessary graph structures within $\gamma$-robust graphs, which will help us establish the lower bounds on edge counts.

\subsection{Case 1: Graphs with Odd $n$}
We now examine the lower bounds on the number of edges for the $\gamma$-robust graphs with odd $n$. To this end, we first study the smallest maximum clique that a $\gamma$-robust graph with odd $n$ must contain in order to satisfy the conditions of $r$-robustness (Definition~\ref{def:r_robust}).

\begin{figure*}
    \centering
\includegraphics[width=1\linewidth]{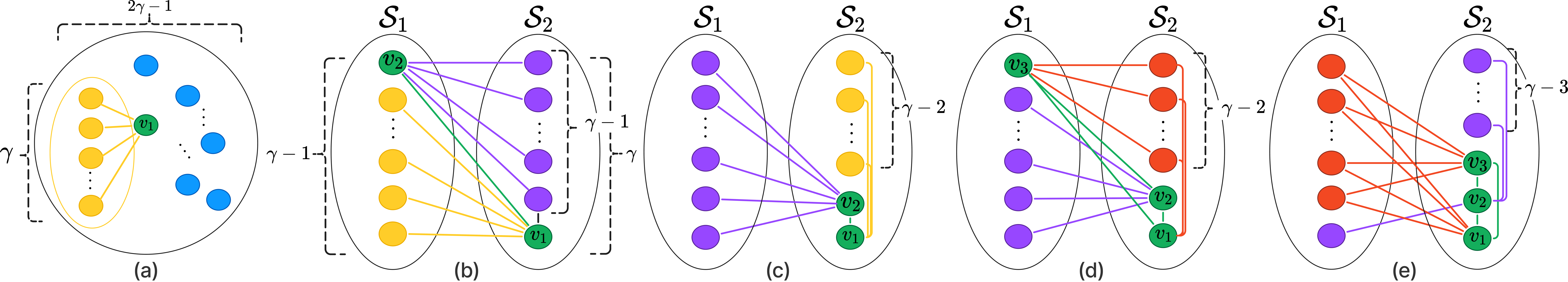}
    \caption{These figures illustrate the first two iterations in the proof of~\Cref{lem:complete}. Figure~(a) shows the initial configuration, where node $v_1\in \mathcal V$ has $\gamma$ neighbors (colored in yellow). Figure~(b) shows the first iteration where node $v_2 \in \mathcal S_1$ is connected to $\gamma-1$ purple nodes and to $v_1$ in $\mathcal S_2$. The green nodes $v_1$ and $v_2$ form a $2$-clique. Note we maintain $|\mathcal S_1|=\gamma-1$ and $|\mathcal S_2|=\gamma$ from figures (b) to (e). Figure~(c) shows $\mathcal S_1$ and $\mathcal S_2$ after the completion of the first step, where $\gamma-1$ nodes in $\mathcal S_1$ (including $v_2$) have been swapped with $\gamma-1$ purple nodes in $\mathcal S_2\setminus\{v_1\}$. Then, the start of the second step is shown in Figure~(d), where another node $v_3 \in \mathcal S_1$ has edges with $\gamma-2$ red nodes and with $v_1,v_2$ in $\mathcal S_2$. Here, the green nodes $v_1,v_2,$ and $v_3$ form a $3$-clique. Figure~(e) shows $\mathcal S_1$ and $\mathcal S_2$ after $\gamma-2$ nodes in $\mathcal S_1$ (including $v_3$) have been swapped with $\gamma-2$ red nodes in $\mathcal S_2\setminus\{v_1,v_2\}$. This process continues until a $(\gamma+1)$-clique is formed.}
    \label{fig:complete_graph_proof}
\end{figure*}

\begin{lemma}
\label{lem:complete}
    Let $\mathcal G = (\mathcal V, \mathcal{E})$ be a graph of $n$ nodes where $n$ is odd, and let $\gamma=\lceil \frac n 2 \rceil$. Then, if $\mathcal G$ is $\gamma$-robust, it must contain a $(\gamma+1)$-clique.
\end{lemma}

\begin{proof}
    In this proof, we will show our argument by constructing two subsets $\mathcal S_1,\mathcal S_2\subset \mathcal V$ such that $\mathcal G$ contains a $(\gamma+1)$-clique. Since $\mathcal G$ is $\gamma$-robust, it holds that for any pair of disjoint, non-empty subsets $\mathcal S_1$, $\mathcal S_2$ of $\mathcal V$, at least one of them is $\gamma$-reachable. Hence, by Definition~\ref{def:r_robust} and~\cite[Lemma~5]{LeBlanc13}, for any node in $\mathcal V=\{v_1,\dots, v_n\}$, $|\mathcal N_{v_i}|\geq \gamma$ must hold. WLOG, let $v_1\in \mathcal V$ be a node, as shown in Figure~\ref{fig:complete_graph_proof} (a). Then because $n=2\gamma-1$, we can construct $\mathcal S_1$ to be the set of $\gamma-1$ neighbors of $v_1$ and $\mathcal S_2$ to be the set of the remaining $\gamma$ nodes including $v_1$ such that $|\mathcal S_1|=\gamma-1$, $|\mathcal S_2|=\gamma$, and $\mathcal S_1\cap \mathcal S_2=\emptyset$. This enforces $\mathcal S_1$ to be $\gamma$-reachable.

    Since $\mathcal S_1$ is $\gamma$-reachable, $\exists v_2 \in \mathcal S_1$ such that $|\mathcal N_{v_2} \cap\mathcal S_2|=|\mathcal N_{v_2} \setminus \mathcal S_1|\geq \gamma$. Then in the first step, $v_2$ has an edge with $v_1 \in \mathcal S_2$, forming a $2$-clique as shown in Figure~\ref{fig:complete_graph_proof} (b). We then swap $\gamma-1$ nodes in $\mathcal S_1$ including $v_2$ with $\gamma-1$ nodes in $\mathcal S_2\setminus\{v_1\}$ (colored in purple in Figure~\ref{fig:complete_graph_proof}). Then in the second step, since $\mathcal S_1$ is $\gamma$-reachable even after $v_2$ is swapped out, $\exists v_3 \in \mathcal S_1$ such that $|\mathcal N_{v_3} \cap\mathcal S_2|\geq \gamma$. By construction, since $v_1$ and $v_2$ are among the $\gamma$ nodes in $\mathcal S_2$, $v_3$ forms a $3$-clique with $v_1$ and $v_2$, as shown in Figure~\ref{fig:complete_graph_proof} (d). We then swap $\gamma-2$ nodes in $\mathcal S_1$ including $v_3$ with $\gamma-2$ nodes in $\mathcal S_2\setminus \{v_1,v_2\}$ (colored in red in Figure~\ref{fig:complete_graph_proof}). Again, in the third step, since $\mathcal S_1$ is $\gamma$-reachable even after $v_3$ is swapped out, $\exists v_4 \in \mathcal S_1$ that has edges with all $\gamma$ nodes in $\mathcal S_2$. Since $v_1,v_2,v_3 \in \mathcal S_2$ are three of the $\gamma$ nodes to have edges with $v_4$, the nodes $v_1,v_2,v_3,v_4$ form a $4$-clique. We then swap $\gamma-3$ nodes in $\mathcal S_1$ including $v_4$ with $\gamma-3$ nodes in $\mathcal S_2\setminus \{v_1,v_2,v_3\}$, and thus the construction continues.
    
    Likewise, in the $k^{\text{th}}$ step, $\exists v_{k+1} \in \mathcal S_1$ such that $|\mathcal N_{v_{k+1}} \cap\mathcal S_2|\geq \gamma$ and thus forms a $(k+1)$-clique with $v_1,v_2,\cdots, v_{k-1},v_k \in \mathcal S_2$. We then swap $\gamma-k$ nodes in $\mathcal S_1$ including $v_{k+1}$ with $\gamma-k$ nodes in $\mathcal S_2\setminus \{v_1,\cdots,v_k\}$. At the end of $\gamma-1^{\text{th}}$ step, $\mathcal S_2$ contains $v_1,\cdots,v_\gamma$ that form a $\gamma$-clique. Finally, since $\mathcal S_1$ is $\gamma$-reachable, it contains a node $v_{\gamma+1}$ that has edges with all nodes $v_1,\cdots, v_\gamma \in \mathcal S_2$. Then, $\{v_{\gamma+1}\}\cup \mathcal S_2$ forms a ($\gamma+1)$-clique.
\end{proof}

\Cref{lem:complete} states that for a graph with odd $n$ to be $\gamma$-robust, it must contain a clique size of $\gamma+1$ or more nodes. Stemming from this reasoning, we finalize the lower bound on edge counts for the case of odd $n$:

\begin{lemma}
 \label{lem:min_addition}
    Let $\mathcal G = (\mathcal V, \mathcal{E})$ be a graph of $n$ nodes where $n$ is odd, and let $\gamma=\lceil \frac n 2\rceil$. Then, if $\mathcal G$ is $\gamma$-robust, 
\begin{equation}
        \label{eq:min_addition}
        |\mathcal{E}|\geq \frac {3\gamma(\gamma-1)} {2}.
    \end{equation}
\end{lemma}

\begin{proof}
 Since $\mathcal G$ is $\gamma$-robust, it holds that for any pair of disjoint, non-empty subsets $\mathcal S_1$, $\mathcal S_2\subset\mathcal V$, at least one of them is $\gamma$-reachable (by Definition~\ref{def:r_robust}).
 From~\Cref{lem:complete}, we know that $\mathcal V$ must contain a $(\gamma+1)$-clique $\mathcal C$. This means that $\mathcal G$ must have at least as many edges as to construct a $(\gamma+1)$-clique, i.e., $|\mathcal E|\geq \frac {\gamma(\gamma+1)}{2}$. 

 Let $\mathcal C' =  \mathcal V\setminus \mathcal C$. This then implies $|\mathcal C|= \gamma+1$ and $|\mathcal C'|=\gamma-2$. WLOG, let $\mathcal S_1= \mathcal C'$ and $\mathcal S_2= \mathcal C$. Then, $|\mathcal S_1|=\gamma-2$ and $|\mathcal S_2|=\gamma +1$. This implies $\mathcal S_1$ has to be $\gamma$-reachable. In other words, there must exist $v_1 \in \mathcal S_1$ such that $|\mathcal N_{v_1}\cap \mathcal S_2|=|\mathcal N_{v_1}\setminus\mathcal S_1|\geq \gamma$. Therefore, for such $v_1$ to exist, we need $\gamma$ edges that connect $v_1$ with $\gamma$ nodes in $\mathcal S_2$: $|\mathcal E| \geq \frac {\gamma(\gamma+1)} 2 +\gamma$. Note that $v_1\in \mathcal S_1$ is the node that makes $\mathcal S_1$ $\gamma$-reachable. Therefore, moving $v_1$ from $\mathcal S_1$ to $\mathcal S_2$ requires the new $\mathcal S_1$, which is now without $v_1$, to contain another node $v_2$ that has at least $\gamma$ neighbors in $\mathcal S_2$. This again requires at least $\gamma$ additional edges (i.e. $|\mathcal E| \geq \frac {\gamma(\gamma+1)} 2+2\gamma$). We then move $v_2$ from $\mathcal S_1$ to $\mathcal S_2$, which necessitates $\mathcal S_1$ to have another node $v_3$ to have $\gamma$ neighbors in $\mathcal S_2$. 
  
  We continue this process of (i) connecting a node $v_i \in \mathcal S_1$ with any $\gamma$ nodes in $\mathcal S_2$ and (ii) moving $v_i$ from $\mathcal S_1$ to $\mathcal S_2$ in the $i^{\text{th}}$ step until $\mathcal S_1$ becomes empty. Then, it is possible to repeat these two steps $\gamma-2$ times, as we initially had $|\mathcal S_1|=\gamma-2$. In other words, for each $\gamma-2$ pair of $\mathcal{S}_1$ and $\mathcal{S}_2$, we require $\gamma$ additional edges. Since we started from the setting where $\mathcal S_2= \mathcal C$ and forced $|\mathcal E|$ to increase at least a total of $\gamma(\gamma-2)$, $|\mathcal E| \geq \frac {\gamma(\gamma+1) + 2\gamma(\gamma-2)} 2 \Rightarrow |\mathcal E| \geq \frac {3\gamma(\gamma-1)} 2$. 
\end{proof}

Using~\Cref{lem:min_addition}, we conclude that any $\gamma$-robust graphs with odd $n$, i.e., graphs with an odd number of nodes that achieve maximum $r$-robustness, must have $\frac {3\gamma(\gamma-1)}{2}$ or more edges. In fact,~\Cref{thm:max_r} later shows that the bound in the lemma is tight and cannot be further improved. Furthermore, using~\Cref{lem:min_addition}, we derive a useful corollary:
\begin{corollary}
 \label{cor:lowerbound}
    Let $\mathcal G = (\mathcal V, \mathcal{E})$ be an $r$-robust graph, where $r\in \mathbb Z_{>0}$. Then, \begin{equation}
        \label{eq:bound1}
        |\mathcal{E}|\geq \frac {3r(r-1)} {2}.
    \end{equation}
\end{corollary}
\begin{proof}
In general, the more nodes a graph has, the
more edges it must have to maintain a given robustness. From \cite{Guerrero17}, it is shown that $2r-1$ is the smallest number of nodes a graph must have to be $r$-robust. In other words, $r$-robustness is the maximum robustness a graph with $2r-1$ nodes achieves. Thus,~\Cref{lem:min_addition} is restated as: any $r$-robust graph with $2r - 1$ nodes must have at least $\frac{3r(r-1)}{2}$ edges. Consequently, any graph $\mathcal{G}$ with $2r - 1$ or more nodes must have at least $\frac{3r(r-1)}{2}$ edges to be $r$-robust.
\end{proof}

\begin{remark}
\label{remark:cor1}
    Corollary~\ref{cor:lowerbound} reveals that the bound in~\Cref{lem:min_addition} is in fact a necessary condition for any $r$-robust graph of any number of nodes $n$, ensuring that graphs with edges less than $\frac {3r(r-1)} 2$ cannot be $r$-robust. One useful implication of this result is that it provides a quick way to upper-bound a graph’s robustness level based solely on its number of edges. Directly computing $r$-robustness of a graph is generally difficult~\cite{zhang2015, leblanc2013alg}, and while various approximation methods exist~\cite{sample_based_r_computation, heuristic_robustness_computation, learning_robustness_computation}, they often lack formal guarantees or assume some structures on the given graphs.  In fact, the upper bound on robustness derived from this result can be used to accelerate existing methods by reducing their search space.
\end{remark}

\subsection{Case 2: Graphs with Even $n$}

Now we continue our analysis on the edge counts for the graphs with even $n$. Since $n$ is even, $\gamma=\frac n 2$. This implies that, for the same value of $\gamma$, the value of $n$ is one greater when $n$ is even compared to when it is odd. However, the overall approach is similar to the previous case: we will first determine the minimum size of the maximum clique that $\mathcal{G}$ must contain to be $\gamma$-robust.

\begin{lemma}
\label{lem:complete2}
    Let $\mathcal G = (\mathcal V, \mathcal{E})$ be a graph with $n$ nodes where $n$ is even, and let $\gamma=\frac n 2$. Then, if $\mathcal G$ is $\gamma$-robust, it must contain a $\big(\lfloor\frac {\gamma+4}{2}\rfloor\big)$-clique.
\end{lemma}

Due to length, the proof for~\Cref{lem:complete2} is given in the Appendix (Section~\ref{sec:appendix}). We emphasize that the bounds from Lemmas~\ref{lem:complete} and~\ref{lem:complete2} are only analytical lower bounds we are able to show, and may not reflect the actual necessary maximum clique size.
In fact, unlike the analysis for odd $n$, studying $\gamma$-robust graphs with even $n$ involves topological structures beyond the maximum clique we observe in~\Cref{lem:complete2}. The following lemma presents the additional necessary structures.

\begin{lemma}
\label{lem:S}
       Let $\mathcal G = (\mathcal V, \mathcal{E})$ be a graph of $n$ nodes where $n$ is even, and let $\gamma=\frac n 2$. Then, if $\mathcal G$ is $\gamma$-robust, $\mathcal G$ contains an induced subgraph $\mathcal G_S=(\mathcal V_S, \mathcal{E}_S)$ such that $|\mathcal V_S|=\gamma+1$ and $|\mathcal{E}_S|\geq \lfloor{\frac {\gamma^2+2} 2}\rfloor$.
\end{lemma}

\begin{proof} 
Since $\mathcal G$ is $\gamma$-robust, for any pair of disjoint, non-empty subsets $\mathcal S_1$, $\mathcal S_2 \subset \mathcal V$, at least one of them is $\gamma$-reachable. From~\Cref{lem:complete2}, $\mathcal G$ contains a $\big(\lfloor\frac {\gamma+4}{2}\rfloor\big)$-clique $\mathcal C$.
    WLOG, let $\mathcal S_2$ contain $\gamma+1$ nodes including $\mathcal C$, and let $\mathcal S_1$ contain the remaining $\gamma-1$ nodes. Initially, we set $\mathcal V_S=\mathcal C$. Then, $|\mathcal{E}_S|= \big(\frac 1 2 \big) \big(\lfloor\frac {\gamma+4}{2}\rfloor\big) \big(\lfloor\frac {\gamma+2}{2}\rfloor\big)$. Note there are total $n=2\gamma$ nodes. Since $|\mathcal S_1|=\gamma-1$ and $|\mathcal S_2|=\gamma+1$, $\mathcal S_1$ must be $\gamma$-reachable. 

    By~\Cref{lem:complete2}, we have a $2$-clique and $3$-clique when $\gamma=1$ and $\gamma=2$ respectively, satisfying the statement in the lemma for $\gamma=1,2$. 
    
    Now, we consider $\gamma\geq3$. With this setup, since $\mathcal S_1$ is $\gamma$-reachable, $\exists v_1 \in \mathcal S_1$ such that $|\mathcal N_{v_1}\cap \mathcal S_2|=|\mathcal N_{v_1}\setminus \mathcal S_1|\geq \gamma$. We swap $v_1 \in \mathcal S_1$ with any node in $\mathcal S_2\setminus \mathcal V_S$. Note that $|\mathcal V_S|=\lfloor\frac {\gamma+4}{2}\rfloor$ at this point. Now, we add $v_1$ as a vertex of $\mathcal G_S$ (i.e. $\mathcal V_S$ now contains $v_1$). Then, $|\mathcal E_S|$ increases at least by $\lfloor\frac {\gamma+2}{2}\rfloor$, since $v_1$ may have no edge with one of the nodes in $\mathcal V_S$. Since $\mathcal S_1$ is $\gamma$-reachable even without $v_1$, $\exists v_2 \in \mathcal S_1$ such that $|\mathcal N_{v_2}\cap \mathcal S_2|\geq \gamma$.  Note that $|\mathcal V_S|=\lfloor\frac {\gamma+6}{2}\rfloor$ at this point. We swap $v_2 \in \mathcal S_1$ with a node in $\mathcal S_2\setminus \mathcal V_S$ and add $v_2$ as a new vertex of $\mathcal G_S$ (i.e. $\mathcal V_S$ now contains $v_2$). This increases $|\mathcal E_S|$ at least by $\lfloor\frac {\gamma+4}{2}\rfloor$, as $v_2$ may have no edge with one of the nodes in $\mathcal V_S$. Then, since $\mathcal S_1$ is still $\gamma$-reachable, $\exists v_3 \in \mathcal S_1$ such that $|\mathcal N_{v_3}\cap \mathcal S_2|\geq \gamma$. Note that $|\mathcal V_S|=\lfloor\frac {\gamma+8}{2}\rfloor$. Again, we swap $v_3 \in \mathcal S_1$ with a node in $\mathcal S_2\setminus \mathcal V_S$ and add $v_3$ as a new vertex of $\mathcal G_S$. Then, $|\mathcal E_S|$ increases at least by $\lfloor\frac {\gamma+6}{2}\rfloor$. 
   
   This process continues until when $v_p \in \mathcal S_1$, $p=\gamma+1-\lfloor\frac {\gamma+4}{2}\rfloor$, that makes $\mathcal S_1$ $\gamma$-reachable and gets swapped with a node in $\mathcal S_2\setminus \mathcal V_S$ where $\mathcal V_S=\mathcal C\cup\{v_1,v_2,\cdots, v_{p-1}\}$. Here, adding $v_p$ as a new vertex of $\mathcal G_S$ increases $|\mathcal E_S|$ at least by $\gamma-1$. At this point, $\mathcal S_2=\mathcal V_S$ and $|\mathcal V_S|=\gamma+1$. Then, $|\mathcal E_S| \geq \big(\frac 1 2\big)\big(\lfloor\frac {\gamma+4}{2}\rfloor\big) \big(\lfloor\frac {\gamma+2}{2}\rfloor\big) + (\lfloor\frac {\gamma+2}{2}\rfloor) + (\lfloor\frac {\gamma+4}{2}\rfloor) + (\lfloor\frac {\gamma+6}{2}\rfloor) +\cdots + (\gamma-2) +(\gamma-1)$.
    
    We remove the floor functions: if $\gamma$ is even, $|\mathcal E_S| \geq \big(\frac 1 2\big)\big(\frac {\gamma+4}{2}\big)\big(\frac {\gamma+2} 2\big) +  \big(\frac {\gamma+2} 2\big)+ \big(\frac {\gamma+4}{2}\big) + \cdots + (\gamma-2) +(\gamma-1)= \frac {\gamma^2+2} 2$. If $\gamma$ is odd, $|\mathcal E_S| \geq \big(\frac 1 2\big)\big(\frac {\gamma+3}{2}\big)\big(\frac {\gamma+1} 2\big) +  \big(\frac {\gamma+1} 2\big)+ \big(\frac {\gamma+3}{2}\big) + \cdots + (\gamma-2) +(\gamma-1) = \frac {\gamma^2+1} 2$. Thus, we have $|\mathcal V_S|=\gamma+1$ and $|\mathcal E_S| \geq \lfloor{\frac{\gamma^2+2} 2}\rfloor$.
\end{proof}
\Cref{lem:S} effectively provides a necessary induced subgraph structure that $\mathcal G$ with even $n$ must have to be $\gamma$-robust. Now, using this lemma, we present the following:

\begin{lemma}
        \label{lem:min_addition2}
    Let $\mathcal G = (\mathcal V, \mathcal{E})$ be a graph of $n$ nodes where $n$ is even, and let $\gamma=\frac n 2$. Then, if $\mathcal G$ is $\gamma$-robust, 
    \begin{equation}
        \label{eq:min_addition2}
        |\mathcal{E}|\geq
            \left\lfloor{\frac {\gamma(3\gamma-2)+2} 2 }\right\rfloor.
    \end{equation}
\end{lemma}

\begin{proof}
The approach of the proof is identical to that of~\Cref{lem:min_addition}, except for the fact that in this proof, we are starting with an induced subgraph $\mathcal G_S=(\mathcal V_S, \mathcal E_S)$ such that $|\mathcal V_S|=\gamma+1$ and $|\mathcal{E}_S|\geq \lfloor{\frac {\gamma^2+2} 2}\rfloor$ (which comes from~\Cref{lem:S}).

  Let $\mathcal S_2=\mathcal V_S$ and $\mathcal S_1=\mathcal V\setminus \mathcal V_S$, which means $|\mathcal S_1|=\gamma-1$ and $|\mathcal S_2|=\gamma+1$. This enforces $\mathcal S_1$ to be $\gamma$-reachable. Then, $\exists v_1 \in \mathcal S_1$ such that $|\mathcal N_{v_1}\cap \mathcal S_2|=|\mathcal N_{v_1}\setminus \mathcal S_1|\geq \gamma$. This requires at least $\gamma$ edges (i.e. $|\mathcal E| \geq |\mathcal E_S|+\gamma$). We move $v_1$ from $\mathcal S_1$ to $\mathcal S_2$. Then since $\mathcal S_1$ is $\gamma$-reachable even without $v_1$, $\exists v_2 \in \mathcal S_1$ such that $|\mathcal N_{v_2}\cap \mathcal S_2|\geq \gamma$. This requires at least $\gamma$ new edges (i.e. $|\mathcal E| \geq |\mathcal E_S|+2\gamma$). 
    
    We continue this process of (i) drawing edges between $v_i \in \mathcal S_1$ and $\gamma$ nodes in $\mathcal S_2$ and (ii) moving $v_i \in \mathcal S_1$ into $\mathcal S_2$ until $\mathcal S_1$ becomes empty. Then, we can do these steps $\gamma-1$ times, as we initially have $|\mathcal S_1|=\gamma-1$. In other words, for each $\gamma-1$ pair of $\mathcal{S}_1$ and $\mathcal{S}_2$, we require $\gamma$ additional edges. Since we have started from the setting where $\mathcal S_2= \mathcal V_S$, we have $|\mathcal E| \geq |\mathcal E_S| + \gamma(\gamma-1)$. Therefore, if $\gamma$ is even, $|\mathcal E| \geq \frac {\gamma^2+2} 2 +\gamma(\gamma-1)=\frac {\gamma(3\gamma-2)+2} 2$. If $\gamma$ is odd, and $|\mathcal E| \geq \frac {\gamma^2+1} 2 +\gamma(\gamma-1) = \frac {\gamma(3\gamma-2)+1} 2$.
\end{proof}

Notice that, between the bounds established in Lemmas~\ref{lem:min_addition} and~\ref{lem:min_addition2}, the latter requires more edges. This is expected, as for a fixed value of $\gamma$,~\Cref{lem:min_addition2} considers graphs with one additional node, which naturally necessitates extra edges to maintain the same level of robustness.

\section{Construction of $\gamma$-MERGs}
\label{sec:construction1}
In the previous section, we have established the lower bounds on edge counts that graphs of $n$ nodes must satisfy in order to achieve $\gamma$-robustness (where $\gamma=\lceil \frac n 2 \rceil$ is the maximum robustness). Proceeding with the lower bounds, we construct $\gamma$-MERGs, or graphs that provably achieve $\gamma$-robustness with a minimal set of edges. This is formally given in~\Cref{thm:max_r}, as our first main result:

\begin{theorem}
\label{thm:max_r}
 Let $\mathcal G = (\mathcal V, \mathcal{E})$ be a graph with $n$ nodes, and let $\gamma=\lceil \frac n 2 \rceil$.  Then, $\mathcal G$ is a $\gamma$-MERG under the following construction mechanisms:
 \begin{itemize}
     \item If $n$ is odd, let $\mathcal X\subset \mathcal V$ be a set of $\gamma+1$ nodes that form a $(\gamma+1)$-clique. Connect each of the remaining $\gamma-2$ nodes to any $\gamma$ nodes in $\mathcal X$.
     \item If $n$ is even, let $\mathcal X\subset \mathcal V$ be a set of $\gamma$ nodes, each of which is adjacent to all other nodes in $\mathcal V$. Select $\lceil \frac {\gamma-2} 2 \rceil$ disjoint pairs of nodes in $\mathcal X$ and remove the edge between each pair. 
 \end{itemize} 
\end{theorem}
\begin{proof}
To prove that $\mathcal G$ is $\gamma$-MERG, we need to prove two things for both odd and even values of $n$: (i) it is $\gamma$-robust and (ii) its edge count $|\mathcal E|$ equals the lower bounds specified in~\Cref{lem:min_addition} or~\Cref{lem:min_addition2}, depending on the value of $n$.

\textbf{Odd $\mathbf n$: } We first prove that it is $\gamma$-robust. Let $\mathcal X' = \mathcal V\setminus \mathcal X$. Then, $|\mathcal X|=\gamma+1$ and $|\mathcal X'|=\gamma-2$. 
WLOG, let $\mathcal S_1,\mathcal S_2 \subset \mathcal V$ be nonempty subsets such that $\mathcal S_1\cap \mathcal S_2=\emptyset$ and $|\mathcal S_1|\leq|\mathcal S_2|$. There are two cases. 

1) If $\mathcal X \cap \mathcal S_1=\emptyset$, let a node $i \in \mathcal X'\cap \mathcal S_1$. Since $\mathcal X \cap \mathcal S_1=\emptyset$ and $ \mathcal N_i \subset \mathcal X \subseteq \mathcal S_2$ by construction, $|\mathcal N_i\setminus \mathcal S_1|\geq \gamma$, making $\mathcal S_1$ $\gamma$-reachable. 

2) If $\mathcal X \cap \mathcal S_1\neq\emptyset$, denote $|\mathcal X\cap \mathcal S_1|=p\leq \gamma-1$. Then, we have $|\mathcal X'\cap\mathcal S_1|\leq \gamma-1-p$, $|\mathcal X\setminus\mathcal S_1|=\gamma+1-p$, and $p-1\leq|\mathcal X'\setminus\mathcal S_1|\leq \gamma -2$. For $i\in \mathcal X \cap \mathcal S_1$, $|(\mathcal N_i\setminus\mathcal S_1)\cap \mathcal X|=|\mathcal X\setminus \{i\}|-(p-1)=\gamma-(p-1)$, as $\mathcal X$ forms a $(\gamma+1)$-clique. For $p=1$, $\exists i \in \mathcal X \cap \mathcal S_1$ such that $|\mathcal N_i\setminus \mathcal S_1|\geq |\mathcal X \setminus \{i\}|= \gamma$.
For $p\geq 2$, note that by construction, every node in $\mathcal X'$ has edges with any $\gamma$ nodes in $\mathcal X$, where $|\mathcal X|=\gamma+1$. Speaking differently, every node in $\mathcal X'$ does not have an edge with only one node in $\mathcal X$. Thus, all nodes in $\mathcal X' \setminus \mathcal S_1$ have edges with at least $p-1$ nodes in $\mathcal X \cap \mathcal S_1$, because it cannot have an edge with only one node in $\mathcal X \cap \mathcal S_1$. Then, as $p-1\leq|\mathcal X'\setminus \mathcal S_1|\leq \gamma-2$, there are at least a total of $(p-1)(p-1)$ edges between the nodes in $\mathcal X' \setminus \mathcal S_1$ and $\mathcal X \cap \mathcal S_1$. Since $|\mathcal X\cap \mathcal S_1|=p$ and by Pigeonhole Principle, $\exists i \in \mathcal X \cap \mathcal S_1$ such that $|(\mathcal N_i\setminus \mathcal S_1)\cap \mathcal X'|\geq \big\lceil{\frac {(p-1)(p-1)}{p}}\big\rceil= \big\lceil{\frac {p^2-2p+1}{p}}\big\rceil=p-1$. Then, $\exists i \in \mathcal X \cap \mathcal S_1$ such that 
 $|\mathcal N_i\setminus \mathcal S_1|=|(\mathcal N_i\setminus\mathcal S_1)\cap \mathcal X|+|(\mathcal N_i\setminus\mathcal S_1)\cap \mathcal X'|\geq \gamma-(p-1)+p-1=\gamma$. 
Thus, whether $\mathcal S_1$ contains a node in $\mathcal X$ or not, $\mathcal S_1$ is always $\gamma$-reachable.

Now we show that $\mathcal{E}$ is a minimal set. $\mathcal G$ has $\frac {(\gamma+1)(\gamma)} 2$ edges to form a $(\gamma+1)$-clique. Additionally, we introduce $\gamma(\gamma-2)$ edges for connecting each of the $\gamma-2$ nodes in $\mathcal X'$ with $\gamma$ nodes in $\mathcal X$. Summing them up, we have $\frac {3\gamma(\gamma-1)} 2$ which equals the lower bound given in~\Cref{lem:min_addition}.

\textbf{Even $\mathbf n$: } We first prove that it is $\gamma$-robust. Let $\mathcal X' = \mathcal V\setminus \mathcal X$. Then, $|\mathcal X|=\gamma$ and $|\mathcal X'|=\gamma$. WLOG, let $\mathcal S_1$ and $\mathcal S_2$ be nonempty subsets of $\mathcal V$ such that $\mathcal S_1\cap \mathcal S_2=\emptyset$ and $|\mathcal S_1|\leq|\mathcal S_2|$. That means $1\leq|\mathcal S_1|\leq \gamma$. There are two cases. 

1) If $\mathcal X \cap \mathcal S_1 = \emptyset$, let a node $i \in \mathcal X'\cap \mathcal S_1$. Since $\mathcal X \cap \mathcal S_1=\emptyset$ and $\mathcal N_i=\mathcal X$, $|\mathcal N_i \setminus \mathcal S_1|\geq \gamma$, and thus $\mathcal S_1$ is $\gamma$-reachable. 

2) If $\mathcal X \cap \mathcal S_1 \neq \emptyset$, let a node $i \in \mathcal X \cap \mathcal S_1$. Note that after removing an edge among $\lceil \frac {\gamma-2} 2 \rceil$ pairs of nodes in $\mathcal X$, there exists at least one (two in case of even $\gamma$) node(s) that is connected to all other nodes in $\mathcal G$. In other words, $\exists v_* \in \mathcal V$ s.t. $|N_{v_*}|=2\gamma-1$. Now, when $1\leq |\mathcal S_1|\leq \gamma-1$, $|\mathcal N_i \setminus \mathcal S_1|\geq \gamma$, making $\mathcal S_1$ $\gamma$-reachable. In case $|\mathcal S_1|=|\mathcal S_2|=\gamma$, we know there exists $v_*$ in either $\mathcal S_1$ or $\mathcal S_2$ such that $|\mathcal N_{v_*}\setminus \mathcal S_i|\geq \gamma$, making $\mathcal S_i$ $\gamma$-reachable. Thus, whether $\mathcal S_1$ contains a node in $\mathcal X$ or not, either $\mathcal S_1$ or $\mathcal S_2$ is always $\gamma$-reachable.

Now, we examine its edge set. The $\gamma$ nodes in $\mathcal X$ connecting to every other node add up to a total of $\frac {\gamma(\gamma-1)+2\gamma(\gamma)} 2 =\frac {3\gamma^2-\gamma} 2$ edges. Then, if we subtract $\lceil \frac {\gamma-2} 2 \rceil$ from it, we get $\frac {\gamma(3\gamma-2)+2} 2$ for even $\gamma$ and $\frac {\gamma(3\gamma-2)+1} 2$ for odd $\gamma$, which equals the lower bound given in~\Cref{lem:min_addition2}. 
\end{proof}

\begin{figure}
    \centering
\includegraphics[width=1\linewidth]{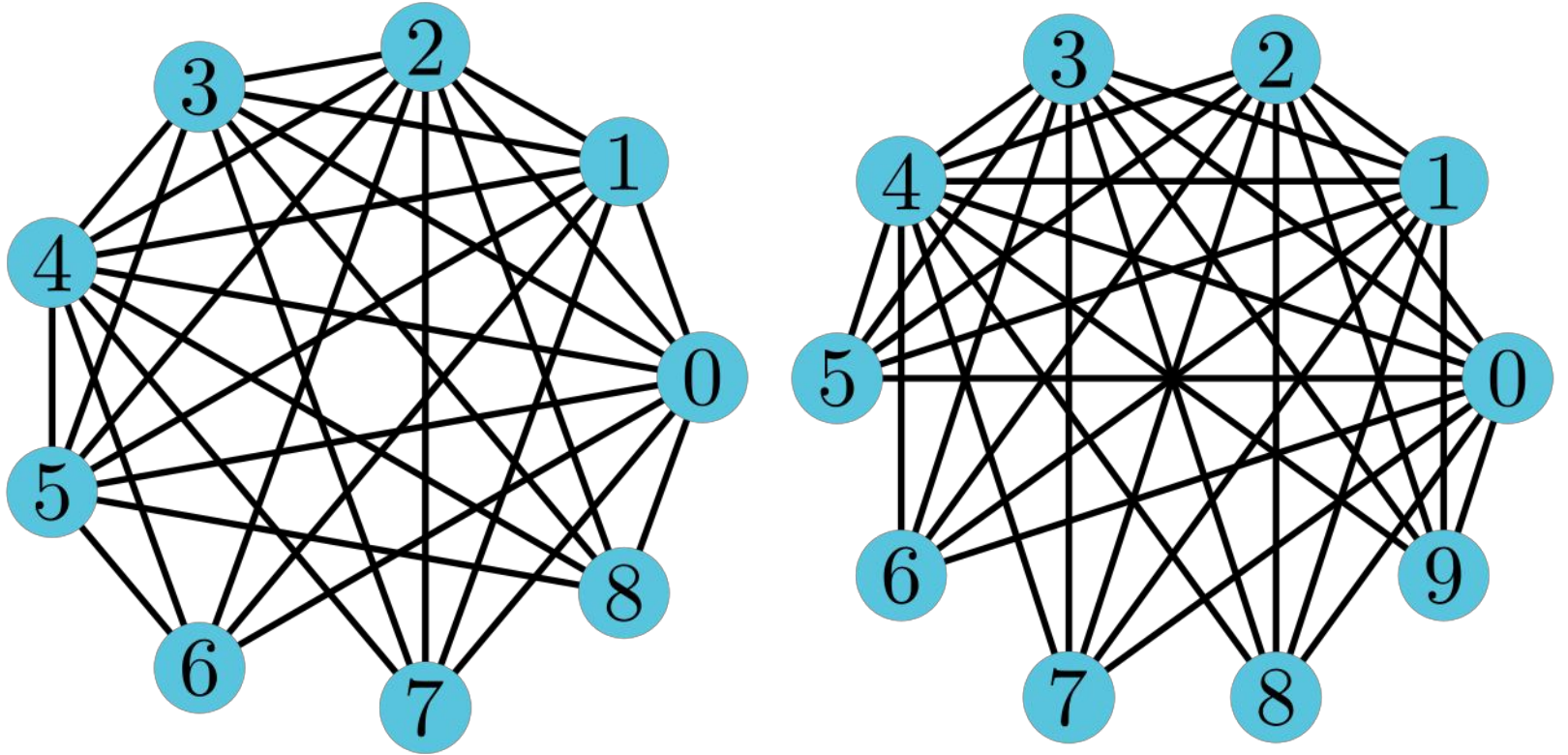}
    \caption{Visualizations of $5$-MERGs with $9$ (left) and $10$ (right) nodes.}
    \label{fig:example1}
\end{figure}

\Cref{thm:max_r} provides a systematic procedure for constructing $\gamma$-MERGs, i.e., graphs that achieve maximum $r$-robustness with a minimal set of edges, thereby solving Problem~\ref{prob:first}. Examples of $\gamma$-MERGs with $n=9$ (left) and $n=10$ (right) are visualized in Figure~\ref{fig:example1}. The theorem also suggests that the bounds given in Lemmas~\ref{lem:min_addition} and \ref{lem:min_addition2} are tight. In fact, removing any edge from a $\gamma$-MERG immediately drops its robustness by at least one, as demonstrated later in Section~\ref{subsub:second_sim}.

\begin{remark}

\label{remark:similiarity}
    We observe that the construction of $F$-elemental graphs in~\cite[Prop.~1]{Guerrero17} bears similarities to our approach in~\Cref{thm:max_r} for odd $n$. Both methods aim to achieve maximum $r$-robustness and include a $\gamma$-clique; the key difference lies in how the remaining non-clique nodes are connected. We also observe that the total number of edges required to achieve maximum robustness for odd $n$ is the same for both approaches (see~\Cref{tab:edge_comparison_single}), suggesting the existence of a more general construction framework for odd $n$. While generalizing these constructions is left for future work, we note that our method additionally handles even $n$ and provides a formal proof of edge-count minimality, which is not addressed in~\cite{Guerrero17}.
\end{remark}

\section{Bounds on Edge Counts for $(\gamma,\gamma)$-MERGs}
\label{sec:lower_bounds2}

In Sections~\ref{sec:lower_bounds1} and~\ref{sec:construction1}, we have examined and constructed $\gamma$-MERGs to address Problem~\ref{prob:first}. In the following two sections, we focus on $(\gamma,\gamma)$-MERGs, graphs with a given number of nodes that achieve maximum $(r,s)$-robustness with a minimal set of edges, thereby addressing Problem~\ref{prob:second}.  

Similar to what we have done in the previous sections, we consider two cases: odd $n$ and even $n$. However, note that $(\gamma,\gamma)$-MERGs for odd $n$ have already been characterized in~\cite[Lemma 4]{LeBlanc13}, which is given below for completeness:
\begin{lemma}\cite{LeBlanc13}
    Let $n>1$ be odd and $\gamma=\lceil \frac n 2 \rceil$. Then, complete graphs are the only graphs with $n$ nodes that are $(\gamma, s)$-robust with $s\geq \lfloor \frac n 2 \rfloor$.
    \label{lem:complete_graph}
\end{lemma}

\Cref{lem:complete_graph} essentially confirms that when $n$ is odd, the graph must be a complete graph to be a $(\gamma,\gamma)$-MERG. This also implies that when $n$ is odd, for a graph to be $(\gamma,\gamma)$-robust, it must have at least $\frac {(2\gamma-1)(2\gamma-2)}{2} = (\gamma-1)(2\gamma-1)$ edges.

\subsection{Graphs with Even $n$}

We now introduce the lower bounds on edge counts for $(\gamma,\gamma)$-robust graphs when $n$ is even. This is formally established in~\Cref{lem:min_addition4}. This lemma will be later used in Section~\ref{sec:construction2} to construct $(\gamma,\gamma)$-MERGs for any $n$. However, before that, we introduce another useful lemma from \cite{LeBlanc13}:
\begin{lemma}
    Given an $(r, s)$-robust graph $\mathcal G = (\mathcal V, \mathcal E)$, with $0 \leq r \leq \lceil \frac n 2\rceil$ and $1 \leq s \leq n$, the
minimum degree of $\mathcal G$, $\delta_{\min}(\mathcal G)$, is at least
\begin{equation}
    \delta_{\min}(\mathcal G) \geq \begin{cases}
        2r-2 & \text{if } s\geq r\\
        r+s-1 & \text{otherwise}
    \end{cases}.
    \label{eq:cases}
\end{equation}
\label{lem:min_degree}
\end{lemma}
Because we consider $r=s=\gamma$ in our paper, what we really care in \eqref{eq:cases} is the inequality for the first case: $\delta_{\min}(\mathcal G)\geq 2r-2$. Therefore, we have the following corollary:
\begin{corollary}
  Given an $(r, r)$-robust graph $\mathcal G = (\mathcal V, \mathcal E)$ of $n$ nodes, the
minimum degree of $\mathcal G$, $\delta_{\min}(\mathcal G)\geq 2(r-1)$. Furthermore, $|\mathcal E| \geq n(r-1)$.
\label{cor:min_degree}
\end{corollary}
\begin{proof}
    As the direct application of~\Cref{lem:min_degree}, we conclude $\delta_{\min}(\mathcal G)\geq 2(r-1)$. Then, each node has at least $2(r-1)$ neighbors. Thus, since $\mathcal G$ is an undirected graph, $|\mathcal E|\geq \frac {2n(r-1)}{2} = n(r-1)$.
\end{proof}

Corollary~\ref{cor:min_degree} provides both the minimum degree required for each node and the total number of edges necessary for a graph to be $(r,r)$-robust. 
Similar to how we have used the notion of cliques as the building blocks of our analyses in Section~\ref{sec:lower_bounds1}, we will use Corollary~\ref{cor:min_degree} to analyze graphs with even $n$. By Corollary~\ref{cor:min_degree}, we know for a graph with even $n$ to be $(\gamma,\gamma)$-MERG, it must have at least $n(\gamma-1)=2\gamma(\gamma-1)$ edges. However, by leveraging the definition of $(r,s)$-robustness (Definition~\ref{def:rs_robust}), we derive a tighter lower bound:
\begin{figure}
    \centering
\includegraphics[width=0.7\columnwidth]{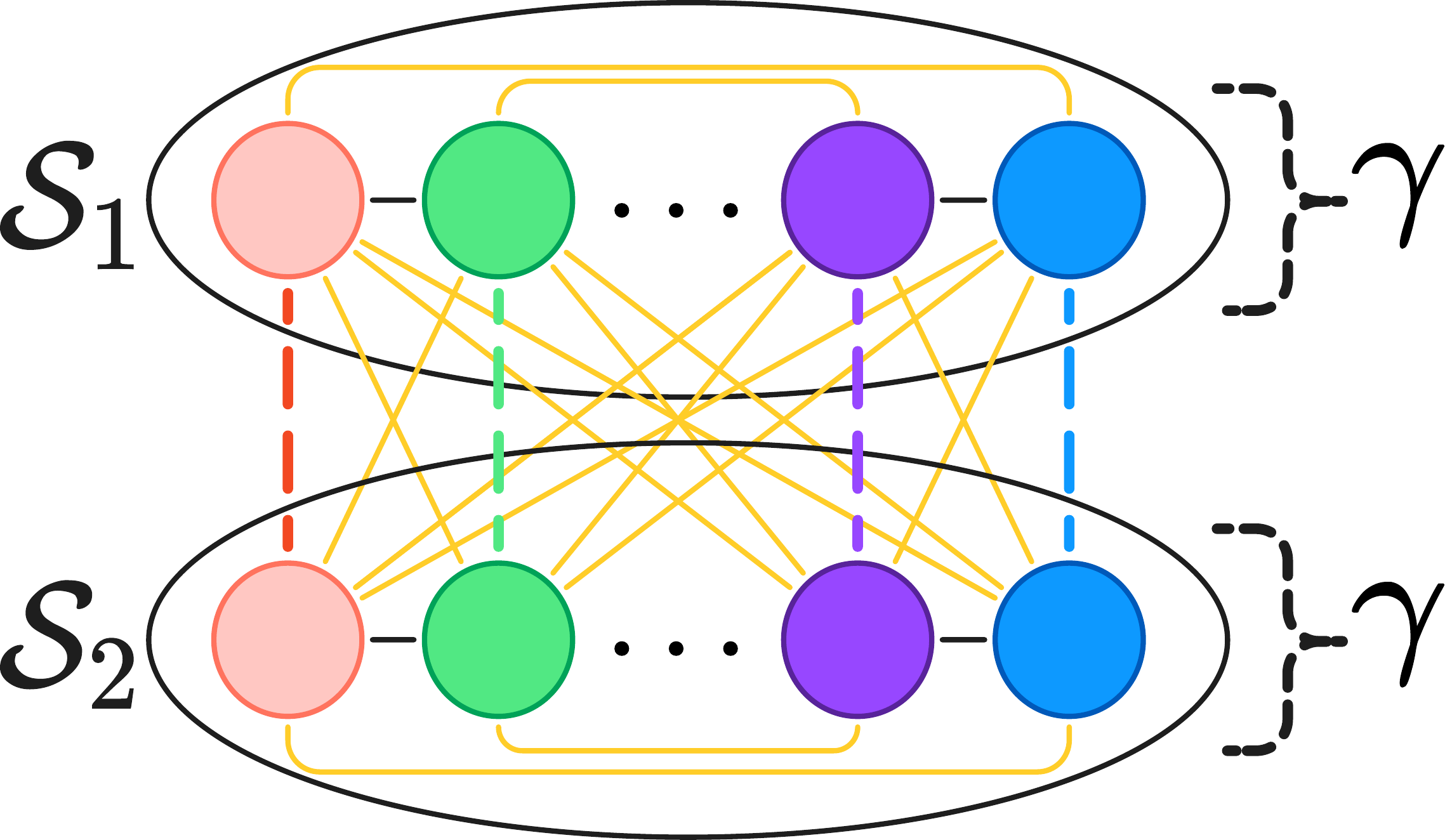}
    \caption{This figure visualizes the contradiction shown in the proof for~\Cref{lem:min_addition4}. The yellow solid lines represent edges, whereas the colored dotted lines represent the absences of edges. There are two disjoint non-empty subsets $\mathcal S_1,\mathcal S_2 \subset \mathcal V$, in which all nodes only have $\gamma-1$ neighbors outside of the subsets, necessitating at least $\lceil \frac {\gamma} 2 \rceil$ of the colored dotted lines to become actual edges to resolve the contradiction.}
    \label{fig:proof}
\end{figure}

\begin{lemma}
   Let $\mathcal G = (\mathcal V, \mathcal{E})$ be a graph of $n$ nodes where $n$ is even, and let $\gamma=\frac n 2 $. Then, if $\mathcal G$ is $(\gamma,\gamma)$-robust,
   \label{lem:min_addition4}
\begin{equation}
    |\mathcal E|\geq 2\gamma(\gamma-1) + \Big\lceil{\frac {\gamma}{2}}\Big\rceil.
    \label{eq:bound4}
\end{equation}
\end{lemma}
\begin{proof}
    By Corollary~\ref{cor:min_degree}, we know $\delta_{\min}(\mathcal G)\geq 2(\gamma-1)$ and $|\mathcal E|\geq (2\gamma)(\gamma-1)$. In the rest of the proof, we will show that some nodes need more than $2(\gamma-1)$ neighbors for $\mathcal G$ to satisfy definition of $(r,s)$-robustness (Definition~\ref{def:rs_robust}). Assume to the contrary that every node $i \in \mathcal V$ only has $2(\gamma-1)$ neighbors and $\mathcal G$ is $(\gamma,\gamma)$-robust. Because (i) there are $2\gamma$ nodes and (ii) each node has $2(\gamma-1)$ neighbors, it is possible to form disjoint $\gamma$ pairs of nodes that do not have an edge together. Now, form disjoint non-empty subsets $\mathcal S_1, \mathcal S_2 \subset \mathcal V$ such that each subset contains one node of each pair that do not have an edge. Since we have $\gamma$ pairs of nodes without an edge, $|\mathcal S_1|=|\mathcal S_2|=\gamma$. By the definition of $(r,s)$-robustness, we need at least $\gamma$ nodes that have $\gamma$ or more neighbors outside of their subsets. However, we have reached a contradiction, as all nodes in $\mathcal S_1$ and $\mathcal S_2$ have only $\gamma-1$ neighbors outside of their subsets. To resolve the contradiction, we need at least $\lceil \frac {\gamma} {2} \rceil$ edges, which are used to create edges between at least $\gamma$ nodes (as shown in Figure~\ref{fig:proof}). Thus, combining everything, we need at least $2\gamma(\gamma-1)+\lceil \frac {\gamma} {2} \rceil$ edges. 
\end{proof}

\Cref{lem:min_addition4} provides the necessary conditions on the edge counts of $(\gamma,\gamma)$-robust graphs with even numbers of nodes. Combining Lemmas~\ref{lem:complete_graph} and~\ref{lem:min_addition4}, we will construct $(\gamma,\gamma)$-MERGs in the next section.

\section{Construction of $(\gamma,\gamma)$-MERGs}
\label{sec:construction2}

\begin{figure}
    \centering
\includegraphics[width=1\linewidth]{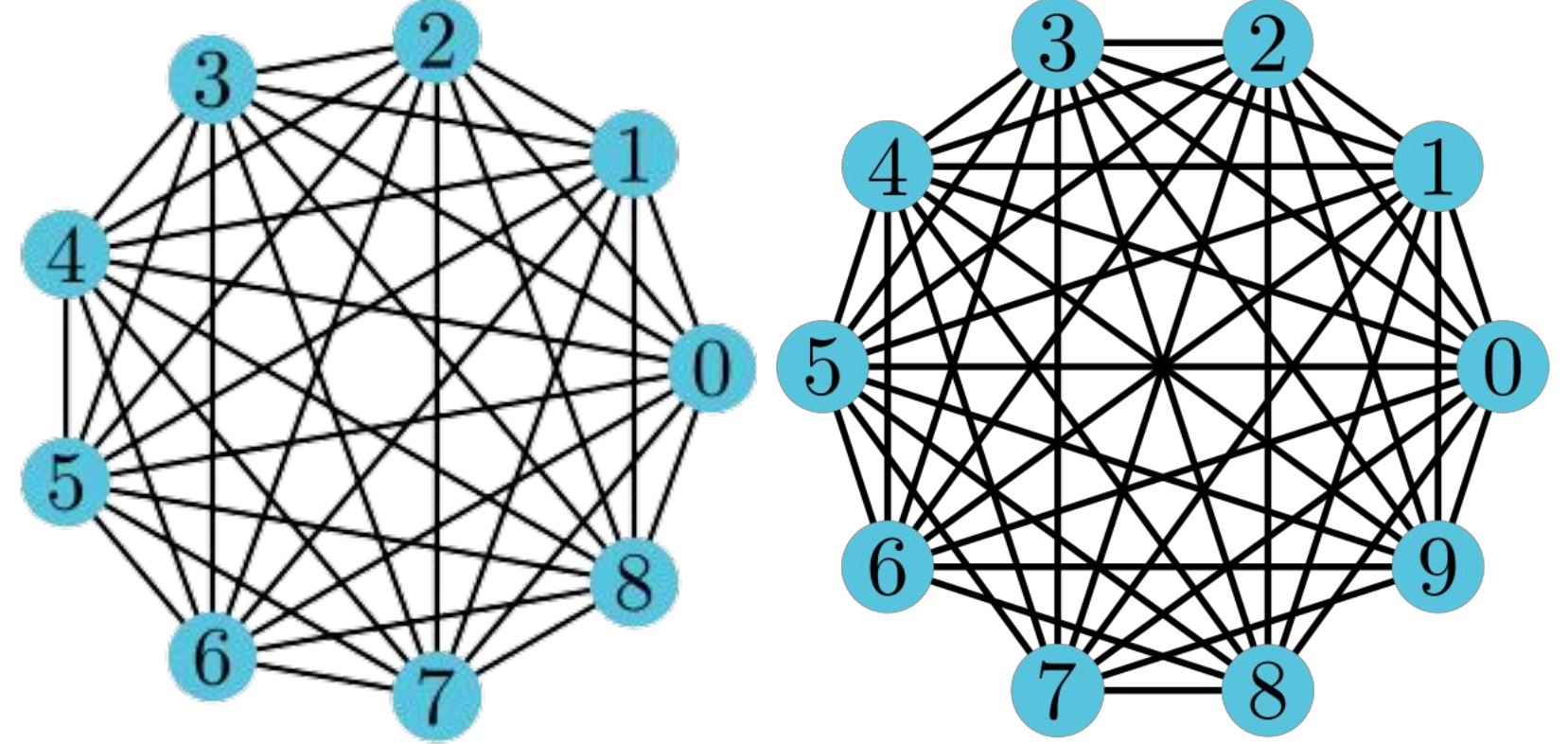}
    \caption{Visualizations of $(5,5)$-MERGs with $9$ (left) and $10$ (right) nodes.}
    \label{fig:example2}
\end{figure}

In the previous section, we presented the lower bounds on edge counts that any graph with $n$ nodes must satisfy in order to achieve $(\gamma,\gamma)$-robustness, where $\gamma=\lceil \frac n 2 \rceil$. Now, we present our second main result by constructing $(\gamma,\gamma)$-MERGs, or graphs that achieve $(\gamma,\gamma)$-robustness with a minimal set of edges:

\begin{theorem}
    Let $\mathcal G = (\mathcal V, \mathcal E)$ be a graph of $n$ nodes, and let $\gamma=\lceil \frac n 2 \rceil$. Let every node have $2(\gamma-1)$ neighbors. If $n$ is even, find $\lceil \frac \gamma 2\rceil$ pairs of nodes without an edge and connect them. Then, $\mathcal G$ is a $(\gamma,\gamma)$-MERG.
    \label{thm:max_rs}
\end{theorem}
\begin{proof}

\textbf{Odd $n$: } Because all $2\gamma-1$ nodes have $2(\gamma-1)$ neighbors, that means $\mathcal G$ forms a complete graph. Therefore by~\Cref{lem:complete_graph}, $\mathcal G$ is a $(\gamma,\gamma)$-MERG when $n$ is odd.

\textbf{Even $\mathbf n$:} To prove that $\mathcal G$ is $(\gamma,\gamma)$-MERG for even values of $n$, we need to prove two things: (i) it is $(\gamma,\gamma)$-robust and (ii) its edge count $|\mathcal E|$ equals to the lower bounds specified in~\Cref{lem:min_addition4}.

First, we prove its robustness. WLOG, we form disjoint non-empty subsets $\mathcal S_1, \mathcal S_2 \subset \mathcal V$ such that $|\mathcal S_1|\leq |\mathcal S_2|$, which implies $1\leq|\mathcal S_1|\leq \gamma$. By construction, there are at least $\gamma$ nodes in $\mathcal G$ that are connected to all other nodes (i.e., they have edges with $2\gamma-1$ nodes). There are two different cases:

1) If $|\mathcal S_1|\leq \gamma-1$, $|\mathcal N_i\setminus \mathcal S_1|\geq |\mathcal N_i| - (\gamma-2)\geq \gamma$ $\forall i\in \mathcal S_1$ since $|\mathcal N_i|\geq 2(\gamma-1)$ $\forall i \in \mathcal V$. Thus, every node in $\mathcal S_1$ has at least $\gamma$ neighbors outside of $\mathcal S_1$, which satisfies the first condition given in Definition~\ref{def:rs_robust}.

2) If $|\mathcal S_1|=|\mathcal S_2|=\gamma$, note that we have at least $\gamma$ nodes that have $2\gamma-1$ neighbors. Thus, however we form $\mathcal S_1$ and $\mathcal S_2$, there will be at least $\gamma=\frac n 2$ nodes in $\mathcal S_1$ and $\mathcal S_2$ combined that have at least $2\gamma-1-(\gamma-1)=\gamma$ neighbors outside of $\mathcal S_1$ and $\mathcal S_2$, respectively. This satisfies the third condition given in Definition~\ref{def:rs_robust}. Since $\mathcal G$ satisfies Definition~\ref{def:rs_robust} with $r=s=\gamma$ in all cases, $\mathcal G$ is $(\gamma,\gamma)$-robust.

Now we show that $\mathcal E$ is a minimal set. We have a total degree of $4\gamma(\gamma-1) + K$, where $K=\gamma+1$ if $\gamma$ is odd and $K=\gamma$ if $\gamma$ is even. Thus, the total number of edges is $|\mathcal E|= \frac {4\gamma(\gamma-1) + K}{2} = \frac {K} {2} + 2\gamma(\gamma-1)$,
which equals the lower bound given in~\Cref{lem:min_addition4}. 
\end{proof}

\Cref{thm:max_rs} addresses Problem~\ref{prob:second} by constructing $(\gamma,\gamma)$-MERGs, graphs that achieve maximum $(r,s)$-robustness with a minimal set of edges. Similar to $\gamma$-MERGs, these graphs also have as many edges as specified by the lower bounds in Lemmas~\ref{lem:complete_graph} and~\ref{lem:min_addition4}. This shows that the lower bounds given in the lemmas are tight. In fact, removing any edge from $(\gamma,\gamma)$-MERGs immediately reduces the graphs' robustness and thus resilience to misbehaving agents.  

Examples of $(\gamma,\gamma)$-MERGs as the results of~\Cref{thm:max_rs} are visualized in Figure~\ref{fig:example2}. The graph on the left has $n=9$ nodes and forms a complete graph to achieve $(5,5)$-robustness, as discussed in~\Cref{lem:complete_graph}. On the contrary, the graph on the right has $n=10$ nodes and is $(5,5)$-robust.

Furthermore, we highlight that the construction mechanisms given in~\Cref{thm:max_rs} are necessary and sufficient in the sense that they encompass all the possible variations of $(\gamma,\gamma)$-MERGs.  
\begin{prop}
    \label{prop:subgraph}
    Let $\mathcal G = (\mathcal V, \mathcal E)$ be a graph of $n$ nodes, and let $\gamma=\lceil \frac n 2 \rceil$. Then, $\mathcal G$ is $(\gamma,\gamma)$-robust if and only if $\mathcal G$ contains a graph $\mathcal G'=(\mathcal V, \mathcal E')$ from~\Cref{thm:max_rs} as a spanning subgraph (i.e., $\mathcal E' \subseteq \mathcal E$).
\end{prop}
\begin{proof}

    (necessary) \textbf{Odd $n$: } It is shown in~\Cref{lem:complete_graph} that complete graph is the only graph with $n$ nodes that is $(\gamma,s)$-robust where $s\geq \lfloor \frac n 2 \rfloor$. This indicates that $\mathcal G$ must be a complete graph, which means $\mathcal E' = \mathcal E$, as $\mathcal G'$ from~\Cref{thm:max_rs} is also a complete graph.

    \textbf{Even $n$: }
    Here we prove the claim by contradiction. Assume by contradiction that $\mathcal G$ is $(\gamma,\gamma)$-robust and $\mathcal E' \not \subseteq \mathcal E$. Let $\mathcal N_i$ and $\mathcal N_i'$ be the neighbor sets of agent $i\in \mathcal V$ in $\mathcal G$ and $\mathcal G'$, respectively. Then, we define $\mathcal M=\{i \in \mathcal V\mid |\mathcal N_i|=2\gamma-1\}$, $\mathcal P=\{i \in \mathcal V\mid |\mathcal N_i|=2\gamma-2\}$, $\mathcal M'=\{i \in \mathcal V\mid |\mathcal N_i'|=2\gamma-1\}$, and $\mathcal P'=\{i \in \mathcal V\mid |\mathcal N_i'|=2\gamma-2\}$.
    
    Note $\mathcal G'$ has $|\mathcal M'|=2\cdot \lceil \frac \gamma 2 \rceil$ nodes with degree of $2\gamma-1$ and $|\mathcal P'|=2\gamma-|\mathcal M'|$ nodes with degree of $2\gamma-2$. By Corollary~\ref{cor:min_degree}, all nodes cannot have less than $2\gamma-2$ neighbors for $\mathcal G$ to be $(\gamma,\gamma)$-robust. Furthermore, they cannot have more than $2\gamma-1$ neighbors. What that means is that the only way for $\mathcal E' \not \subseteq \mathcal E$ is to have $|\mathcal M|<|\mathcal M'|$. Then, $|\mathcal E|<2\gamma(\gamma-1)+\lceil \frac \gamma 2 \rceil$. However, that is a contradiction, as for $\mathcal G$ to be $(\gamma,\gamma)$-robust, $|\mathcal E|\geq 2\gamma(\gamma-1)+\lceil \frac \gamma 2 \rceil$ by~\Cref{lem:min_addition4}. 

    (sufficiency) We know that $\mathcal G'$ is a $(\gamma,\gamma)$-MERG, i.e., $(\gamma,\gamma)$-robust graph with a minimal set of edges. Since $(r,s)$-robustness is monotonic with respect to number of edges~\cite[Lemma~3]{LeBlanc13}, if $\mathcal E' \subseteq \mathcal E$, $\mathcal G$ must be $(\gamma,\gamma)$-robust. 
\end{proof}
It follows from~\Cref{prop:subgraph} that every $(\gamma,\gamma)$-MERG - that is, every $(\gamma,\gamma)$-robust graph with a minimal edge set - can be constructed using the procedure in~\Cref{thm:max_rs}. Consequently, it is possible to use~\Cref{prop:subgraph} to quickly verify whether a given graph $\mathcal G$ is $(\gamma,\gamma)$-robust by checking whether it contains $(\gamma,\gamma)$-MERGs from \Cref{thm:max_rs} as a spanning subgraph. While finding such subgraphs is in general NP-complete, many efficient algorithms exist (e.g.,~\cite{subgraph_matching1,subgraph_matching2}).

\begin{figure*}[ht!]
    \centering
\includegraphics[width=1\linewidth]{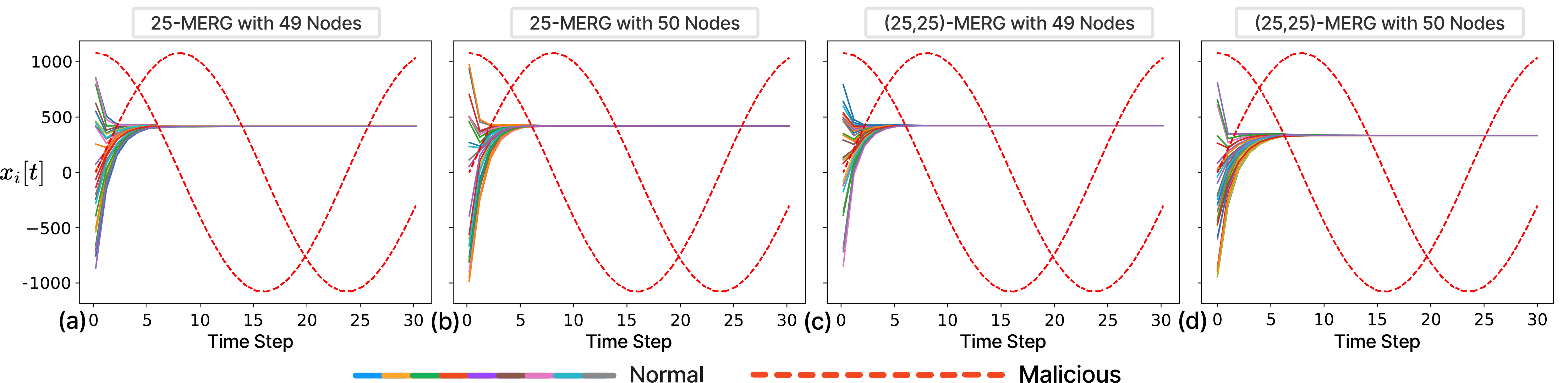}
    \caption{Simulation results for the first set of simulations. We run consensus on $25$-MERGs with (a) $49$ and (b) $50$ nodes, and $(25,25)$-MERGs with (c) $49$ and (d) $50$ nodes, demonstrating resilient consensus through the W-MSR algorithm under $F=12$ and $F=24$ malicious agents, respectively. The normal agents' states are solid colored lines, while malicious agents' states are shown as red dotted lines.}
    \label{fig:consensus}
\end{figure*}

\begin{figure*}[h]
    \centering
\includegraphics[width=1\linewidth]{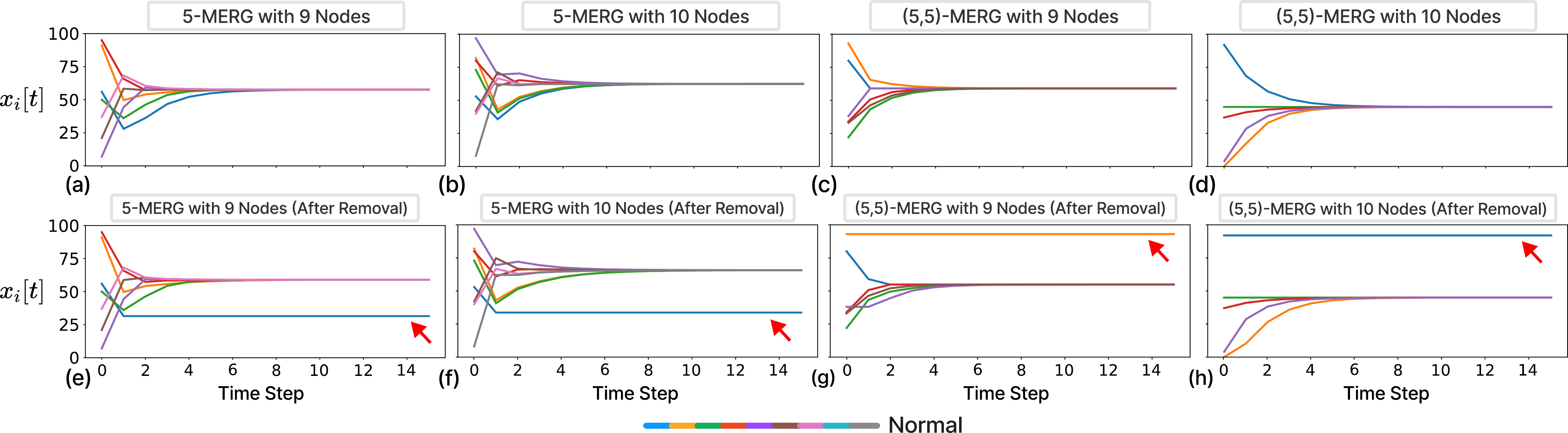}
    \caption{Simulation results for the second set of simulations. Consensus with $F = 2$ Byzantine agents is performed on $5$-MERGs with $n = 9$ (panels (a) and (e)) and $n = 10$ (panels (b) and (f)), both without and with an edge removal, respectively. Similarly, consensus with $F = 4$ malicious agents is executed on $(5,5)$-MERGs with $n = 9$ (panels (c) and (g)) and $n = 10$ (panels (d) and (h)), also without and with an edge removal, respectively. The states of normal agents are shown as solid lines, while those of misbehaving agents are omitted for visual clarity.}
    \label{fig:consensus2}
\end{figure*}

\section{Comparison with Other Graphs}

\label{sec:comparison}
\begin{table}[th]
\centering
\caption{Number of edges required to achieve maximum robustness for different graph construction methods as the number of nodes \(n\) increases.}
\label{tab:edge_comparison_single}
\setlength{\tabcolsep}{4pt}
\begin{tabular}{c c c c}
\toprule
\(n\) & Graphs & Edge Count & Edge Count \\ & &($\gamma$-robustness) & ($(\gamma,\gamma)$-robustness) \\
\midrule
\multirow{3}{*}{5}
& $k$-Circulant~\cite{pirani2023} & 10 & 10\\
& $F$-Elemental~\cite{Guerrero17} & 9 & - \\
& Ours &  9 & 10\\
\midrule
\multirow{3}{*}{30}
& $k$-Circulant~\cite{pirani2023}  & 435 & 435 \\
& $F$-Elemental~\cite{Guerrero17} & - & -\\
& Ours &  323 & 428\\
\midrule
\multirow{3}{*}{105}
& $k$-Circulant~\cite{pirani2023}  & 5460 & 5460\\
& $F$-Elemental~\cite{Guerrero17} & 4134 & - \\
& Ours & 4134 & 5460\\
\midrule
\multirow{3}{*}{300}
& $k$-Circulant~\cite{pirani2023}  & 44850 & 44850 \\
& $F$-Elemental~\cite{Guerrero17} & - & -\\
& Ours & 33601 & 44775\\
\bottomrule
\end{tabular}
\end{table}

In this section, we compare $\gamma$- and $(\gamma,\gamma)$-MERGs from Theorems~\ref{thm:max_r} and~\ref{thm:max_rs} with other existing robust graphs in terms of edge counts. The goal is to highlight that our construction methods require fewer edges than other existing methods. We focus on two classes of graphs for comparison: undirected $k$-circulant graphs~\cite{pirani2023, zhang2015} and $F$-elemental graphs~\cite{Guerrero17}. Table~\ref{tab:edge_comparison_single} summarizes the number of edges required to achieve maximum robustness for each class of graphs as the number of nodes $n$ increases.

An undirected $k$-circulant graph is $\lceil k/2 \rceil$-robust and at least $(\lfloor (k+2)/2 \rfloor, \lfloor (k+2)/2 \rfloor)$-robust for even $k$ and $(\lfloor (k+1)/2 \rfloor, \lfloor (k+1)/2 \rfloor)$-robust for odd $k$~\cite{pirani2023, usevitch2017circulant}. Consequently, achieving the maximum robustness level requires the graph to be fully connected, leading to a high number of edges, as shown in Table~\ref{tab:edge_comparison_single}. On the other hand, $F$-Elemental graphs constructed with~\cite[Prop. 1]{Guerrero17} (discussed in~\Cref{remark:similiarity}) tend to use as few edges as ours. However, these graphs are defined only for $r$-robust graphs with odd $n$ and provide no analysis of edge minimality. In contrast, as supported by our analyses in Sections~\ref{sec:lower_bounds1} and~\ref{sec:lower_bounds2}, our graphs constructed according to Theorems~\ref{thm:max_r} and~\ref{thm:max_rs} consistently use the fewest edges for any $n$ and for both $r$- and $(r,s)$-robustness.

\section{Simulations}
\label{sec:sim}
In the previous sections, we have established systematic ways to construct $\gamma$- and $(\gamma,\gamma)$-MERGs, which provably achieve maximum $r$- and $(r,s)$-robustness with the fewest edges possible. We now validate the correctness of our results through two sets of simulations (presented in Section~\ref{subsub:first_sim} and Section~\ref{subsub:second_sim}). While $\gamma$- and $(\gamma,\gamma)$-MERGs support a variety of resilient algorithms, including estimation, optimization, and learning~\cite{liewi2021byzantine, sundaram2019distributed_opt, xie2021towards} as well as different resilient consensus algorithms~\cite{yuan2021, saldana2017, koushkbaghi2024byzantine, wang2025resilience}, here we focus specifically on resilient consensus via the W-MSR algorithm.

For all simulations, normal agents in the network execute the W-MSR algorithm to achieve consensus despite the presence of misbehaving agents, which transmit adversarial values to disrupt consensus. Through the W-MSR algorithm, normal agents’ states are guaranteed to converge to a common value within the convex hull of their initial conditions, \textit{irrespective of the behaviors of the misbehaving agents}~\cite{LeBlanc13}. The results are shown in Figures~\ref{fig:consensus} and~\ref{fig:consensus2}.

\subsection{Consensus Performance\label{subsub:first_sim}}
In the first set of simulations, we demonstrate that in the presence of $F$ malicious agents, indexed by $i\in\{0,\dots, F-1\}$, normal agents $i\in \{F,\dots, n-1\}$ in $\gamma$- and $(\gamma,\gamma)$-MERGs are guaranteed to achieve resilient consensus. Each normal agent $i$'s initial state $x_i[0] \in \mathbb R$ is randomly generated on the interval $[-1000,1000]$. At time step $t\geq 0$, a malicious agent $i$ transmits an adversarial value: $1080\cos(t/5)$ if $i$ is even, and $1080\sin(t/5)$ if $i$ is odd. 

We consider networks with $n=49$ and $n=50$ agents in all simulations. The maximum $r$-robustness any network with $49$ and $50$ agents achieves is $r=\lceil \frac n 2 \rceil = 25$-robustness. Thus, using~\Cref{thm:max_r}, we construct $25$-MERGs with $49$ and $50$ nodes. Note these graphs achieve $25$-robustness with the fewest edges possible for a given number of nodes, and they tolerate at most $F=12$-local malicious agents. Thus, we run consensus simulations on these network graphs with $12$ malicious agents, and the results are shown in Figure~\ref{fig:consensus} (a) and (b). Similarly, the maximum possible $(r,s)$-robustness any network with $49$ and $50$ agents achieves is $(25,25)$-robustness. Thus, we use~\Cref{thm:max_rs} to construct $(25,25)$-MERGs with $49$ and $50$ agents. These graphs tolerate up to $F= 24$-total malicious agents. We run consensus simulations on these network graphs with $24$ malicious agents, and the results are shown in Figure~\ref{fig:consensus} (c) and (d). As shown in Figure~\ref{fig:consensus}, both $\gamma$- and $(\gamma,\gamma)$-MERGs allow all normal agents to successfully achieve consensus within the convex hull of their initial values through the W-MSR algorithm with $F$ malicious agents.

\subsection{Consensus with an Edge Removal\label{subsub:second_sim}}

These $\gamma$- and $(\gamma,\gamma)$-MERGs achieve their respective maximum robustness levels using the fewest edges possible for a given number of nodes. In other words, removing any edge reduces their robustness levels at least by one. To validate this minimalistic property, we compare the consensus performance of the original $\gamma$- and $(\gamma,\gamma)$-MERGs with that of modified versions in which a single edge has been removed. For each graph, comparison is done using identical initial states for the normal agents and in the presence of $F$ misbehaving agents.

\subsubsection{Simulations with $\gamma$-MERGs} We construct $5$-MERGs with $n=9$ and $n=10$ nodes as shown in Figure~\ref{fig:example1}. Note these graphs achieve $5$-robustness with the fewest edges possible for a given number of nodes, and they tolerate at most $F=2$-local Byzantine agents. We set agents $i \in \{0,1\}$ as Byzantine, and agents $i \in\{2,\dots, n-1\}$ as normal. Each normal agent $i$ is initialized with a state $x_i[0] \in \mathbb Z_{\geq 0}$ sampled uniformly at random from the interval $[0, 49]$ for $i \in \{2, 3, 4\}$ and from $[50, 100]$ for $i \in\{5,\dots, n-1\}$. For all $t\geq 0$, all Byzantine agents adversarially share a value of $100$ with agents $i=0,\dots, \lceil \frac n 2\rceil$ and of $-100$ with other remaining agents.

Figure~\ref{fig:consensus2}~(a)-(b) and (e)-(f) illustrate the consensus results without and with edge removal. We have removed the edge $(3,8)$ for $n=9$ and $(4,9)$ for $n=10$ from the $5$-MERGs shown in Figure~\ref{fig:example1}. In the $5$-MERGs, normal agents successfully reach resilient consensus (shown in Figure~\ref{fig:consensus2}~(a) and~(b)). However, removing the edge from $5$-MERG reduces its robustness to $r = 4$, which is insufficient to guarantee resilient consensus in the presence of $F = 2$ Byzantine agents. Thus, normal agents in the new graphs fail to reach resilient consensus (as shown by the red arrows in Figure~\ref{fig:consensus2}~(e) and~(f)).

\subsubsection{Simulations with $(\gamma,\gamma)$-MERGs}
We use~\Cref{thm:max_rs} to construct $(5,5)$-MERGs with $n=9$ and $n=10$ agents, as shown in Figure~\ref{fig:example2}. These graphs tolerate up to $F=4$-total malicious agents. We set agents $i\in \{0,\dots, 3\}$ as malicious, and agents $i\in\{4,\dots, n-1\}$ as normal. Each normal agent $i$ is initialized with a state $x_i[0] \in \mathbb{R}$ sampled uniformly at random from the interval from $[0, 49]$ for agents $i\in\{4,\dots, 7\}$ and from $[50, 100]$ for the other normal agents.  Malicious agents share adversarial values: agent $3$ transmits $0$ to its neighbors at all times, while the remaining malicious agents transmit $100$ to their neighbors at all times.

Figure~\ref{fig:consensus2}~(c)-(d) and~(g)-(h) illustrate the results without and with edge removal. We have removed the edge $(7,8)$ for $n=9$ and $(5,8)$ for $n=10$ from the $(5,5)$-MERGs visualized in Figure~\ref{fig:example2}. In these $(5,5)$-MERGs, normal agents successfully reach resilient consensus (shown in Figure~\ref{fig:consensus2}~(c) and~(d)). However, removing the edge from each $(5,5)$-MERG reduces its robustness to at most $(5,4)$-robustness, which is insufficient to guarantee resilient consensus in the presence of $F = 4$ malicious agents. As a result, normal agents fail to reach resilient consensus in the modified graphs (as shown by the red arrows in Figure~\ref{fig:consensus2}~(g) and~(h)).

\section{Conclusion}
\label{sec:conc}

This paper investigates and presents the necessary conditions and topological structures of $r$- and $(r,s)$-robust graphs. These results are derived purely from the definitions of robustness, providing a general understanding of what these robustness levels require. We establish the tight lower bounds on the edge counts for graphs with any number of nodes to achieve maximum robustness. Then, we develop further from these results by constructing two classes of graphs of any number of nodes that achieve maximum $r$- and $(r,s)$-robustness with minimal sets of edges. For future work, we aim to extend our results to arbitrary robustness levels.

\section{Appendix}
\label{sec:appendix}

Here, the proof for~\Cref{lem:complete2} is provided below, accompanied by Figure~\ref{fig:complete2} for illustrative support.

\begin{figure*}
    \centering
\includegraphics[width=1\linewidth]{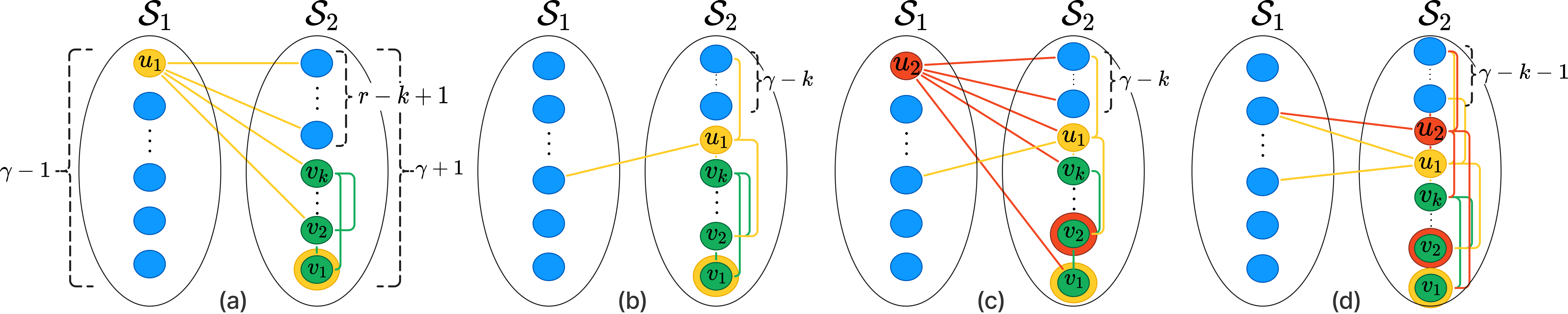}
    \caption{These figures demonstrate the first two iterations of the first part in the proof of~\Cref{lem:complete2}. Throughout figures (a) to (d), we maintain $|\mathcal S_1|=\gamma-1$ and $|\mathcal S_2|=\gamma+1$. Figures (a) and (b) illustrate the first iteration. Here, $\mathcal S_2$ contains $\gamma+1$ nodes including $k$-clique $\mathcal C_{k,1}=\{v_1,\cdots,v_k\}$ (colored in green), while $\mathcal S_1$ contains the remaining $\gamma-1$ nodes. Since $\mathcal S_1$ has to be $\gamma$-reachable, there exists a node $u_1 \in \mathcal S_1$ that has edges with at least $\gamma$ nodes in $\mathcal S_2$. Figure~(a) shows the case scenario where the yellow node $u_1$ does not have an edge with $v_1$ (highlighted in yellow). Then, $u_1$ forms a $k$-clique $\mathcal C_{k,2}$ with $\mathcal  C_{k,1}\setminus\{v_1\}$. Consequently, in any future iteration, if any node $u_i \in \mathcal S_1$, $i\neq 1$, connects to at least $\gamma$ nodes in $\mathcal S_2$ but not to $v_1$, it will form a $(k+1)$-clique with $\mathcal C_{k,2}$. Then, $u_1$ gets swapped with a node in $S_2\setminus \mathcal C_{k,1}$, as shown in Figure~(b). Figures (c) and (d) present the second iteration. Figure~(c) presents the case scenario where $u_2$ only has $\gamma$ edges and does not have an edge with $v_2$ (highlighted in red). As before, if any node $u_i \in \mathcal S_1$, $i\neq 2$, connects to $\gamma$ nodes while missing an edge to $v_2$ in future iteration, it will form a $(k+1)$-clique. Then, $u_2$ gets swapped with a node in $\mathcal S_2\setminus U'$ where $U'= \mathcal C_{k,1}\cup \mathcal C_{k,2}$, as shown in Figure~(d). The process continues, maintaining $\gamma$-reachability of $\mathcal S_1$, until a $\big(\lfloor\frac {\gamma+3}{2}\rfloor\big)$-clique is formed.
    }
    \label{fig:complete2}
\end{figure*}

\begin{proof}
Since $\mathcal G$ is $\gamma$-robust, it holds that for any pair of disjoint, non-empty subsets $\mathcal S_1$, $\mathcal S_2\subset\mathcal V$, at least one of them is $\gamma$-reachable. The proof is divided into two parts. In the first part, we will inductively show our argument by constructing two subsets $\mathcal S_1, \mathcal S_2 \subset \mathcal V$ such that $\mathcal G$ is $\gamma$-robust and contains a $c$-clique where $2 \leq c \leq \lfloor\frac {\gamma+3}{2}\rfloor$. Then, in the second part of the proof, we will show that $\mathcal G$ must contain a $\big(\lfloor\frac {\gamma+4}{2}\rfloor\big)$-clique through some reasoning.

\textbf{First Part:} For the induction step, we fix $|\mathcal S_1|=\gamma-1$ and $|\mathcal S_2|=\gamma+1$, enforcing $\mathcal S_1$ to be $\gamma$-reachable. Hence we have the base case of a $2$-clique. Now, we will iteratively swap nodes between $\mathcal S_1$ and $\mathcal S_2$ to argue some of the edges are necessary for $\mathcal G$ to satisfy Definition~\ref{def:r_robust}.

WLOG, let $\mathcal S_2$ contain $\gamma+1$ nodes including $k$-clique $\mathcal C_{k,1}=\{v_1,\cdots,v_k\}$ (green in Figure~\ref{fig:complete2}), where $k\geq 2$, and let $\mathcal S_1$ contain the remaining $\gamma-1$ nodes. We also denote $U_1 = U_1' = \mathcal C_{k,1}$. Since $\mathcal S_1$ is $\gamma$-reachable, $\exists u_1 \in \mathcal S_1$ (yellow in Figure~\ref{fig:complete2}) such that $|\mathcal N_{u_1}\cap \mathcal S_2|=|\mathcal N_{u_1}\setminus \mathcal S_1|\geq \gamma$. There are two cases: 

\begin{enumerate}
    \item $u_1$ has edges with all $\gamma+1$ nodes in $\mathcal S_2$, or $\gamma$ nodes but not with a node in $\mathcal S_1 \setminus U_1$.
    \item $u_1$ has edges with $\gamma$ nodes but not with $v_1 \in U_1$ (highlighted in yellow in Figure~\ref{fig:complete2}).
\end{enumerate}

If 1), $u_1$ forms a $(k+1)$-clique with $\mathcal C_{k,1}$. If 2), $u_1$ forms another $k$-clique $\mathcal C_{k,2}$ with $U_1' \setminus\{v_1\}$. Note that $u_1$ has edges with all nodes in $\mathcal S_2\setminus\{v_1\}$. Now, if 1), we have a $(k+1)$-clique and done. If 2), let $U_2= \mathcal C_{k,1} \cap \mathcal C_{k,2} = \{v_2,\cdots,v_k\}$ and $U_2'=\mathcal C_{k,1} \cup \mathcal C_{k,2} =\{v_1,\cdots,v_k,u_1\}$. We also swap $u_1$ with any of the nodes in $\mathcal S_2\setminus U_1'$ to get $\mathcal S_1$ and $\mathcal S_2$ shown in Figure~\ref{fig:complete2} (b).

 Then, since $\mathcal S_1$ has to be $\gamma$-reachable, $\exists u_2 \in \mathcal S_1$ (red in Figure~\ref{fig:complete2}) such that $|\mathcal N_{u_2}\cap \mathcal S_2|\geq \gamma$. Again, there are two cases: 
 \begin{enumerate}
     \item $u_2$ has edges with all $\gamma+1$ nodes in $\mathcal S_2$, or $\gamma$ nodes but not with a node in $\mathcal S_2\setminus U_2$
     \item $u_2$ has edges with $\gamma$ nodes but not with $v_2 \in U_2$ (highlighted in red in Figure~\ref{fig:complete2}).
 \end{enumerate}
 
 If 1), $u_2$ forms a $(k+1)$-clique with $\mathcal C_{k,j}$ for some $j\in \{1,2\}$. If 2), $u_2$ forms another $k$-clique $\mathcal C_{k,3}$ with $U_2'\setminus\{v_1,v_2\}=\{v_3,\cdots,v_k,u_1\}$. Also note that $u_2$ has edges with all nodes in $\mathcal S_2\setminus \{v_2\}$ including $u_1$. Now, if 1), we have a $(k+1)$-clique and done. If 2), we update $U_3=\bigcap\limits_{j=1}^3 \mathcal C_{k,j}=\{v_3,\cdots,v_k\}$ and $U_3'=\bigcup\limits_{j=1}^3 \mathcal C_{k,j}=\{v_1,\cdots,v_k,u_1,u_2\}$. We then swap $u_2$ with any node in $\mathcal S_2\setminus U_2'$, continuing the process.
 
Likewise, unless we get case 1), we keep encountering the second case illustrated in the previous steps. To generalize the procedure, let $u_q \in \mathcal S_1$ such that $|\mathcal N_{u_q}\cap \mathcal S_2|\geq \gamma$, $q\in\{1,\cdots,k\}$. Then, $U_q=\bigcap\limits_{j=1}^{q}\mathcal C_{k,j}=\{v_q,\cdots,v_k\}$ and $U_q'=\bigcup\limits_{j=1}^q\mathcal C_{k,j}=\{v_1,\cdots,v_k,u_1,\cdots,u_{q-1}$\}. There are two cases: 
\begin{enumerate}
    \item $u_q$ has edges with all $\gamma+1$ nodes in $\mathcal S_2$, or $\gamma$ nodes but not with a node in $\mathcal S_2\setminus U_q$.
    \item $u_q$ has edges with $\gamma$ nodes in $\mathcal S_2$ but not with $v_q \in U_q$.
\end{enumerate}

 If 1), $u_q$ forms a $(k+1)$-clique with $\mathcal C_{k,j}$ for any $j\in \{1,\cdots, q\}$. If 2), $u_q\in \mathcal S_1$ forms another $k$-clique $\mathcal C_{k,q+1}$ with $U_q'\setminus \{v_1,\cdots, v_q\}$. Now, if 1), we have a $(k+1)$-clique and done. If 2) we swap $u_q$ with any node in $\mathcal S_2\setminus U_q'$ to continue the procedure. 
    
If we continue encountering the case 2) for $k$ times, since $\mathcal S_1$ is $\gamma$-reachable, $\exists u_{k+1} \in \mathcal S_1$ such that $|\mathcal N_{u_{k+1}}\cap \mathcal S_2|\geq \gamma$. Then, there are two cases: 
\begin{enumerate}
    \item $u_{k+1}$ has edges with all $\gamma+1$ nodes in $\mathcal S_2$.
    \item $u_{k+1}$ has edges with $\gamma$ nodes in $\mathcal S_2$ but not with a node in $\mathcal S_2\setminus U_{k+1}$ where $U_{k+1}=\bigcap\limits_{j=1}^{k+1}\mathcal C_{k,j}=\emptyset$.
\end{enumerate}
    
In both cases, $u_{k+1}$ forms a $(k+1)$-clique with $\mathcal C_{k,j}$ for any $j \in \{1,\cdots, k+1\}$, showing a $(k+1)$-clique must exist. Also note that $U_{k+1}'=\{v_1,\cdots,v_k,u_1,\cdots,u_{k}\} \subseteq S_2$. Since (i) at worst $u_1,\cdots, u_k$ as well as $v_1,\cdots,v_k$ need to be in $\mathcal S_2$ for a $(k+1)$-clique to be formed and (ii) $|\mathcal S_2|=\gamma+1$, $k\leq \lfloor\frac {\gamma+1}{2}\rfloor$. That means $\mathcal G$ contains a $\big(\lfloor\frac {\gamma+3}{2}\rfloor\big)$-clique, which completes the first part of the proof.

\textbf{Second Part:} Continuing from the previous paragraph, at $k= \lfloor\frac {\gamma+1}{2}\rfloor$ and for even values of $\gamma$, at worst the case 2) mentioned above is repeated $\frac \gamma 2$ times. Then, $\mathcal S_2$ contains $\mathcal K=\{u_1,\cdots, u_{\frac \gamma 2}\}$, $\mathcal C_{k,1}=\{v_1,\cdots, v_{\frac \gamma 2}$\}, and a node $p$ not in $\mathcal K$ and $\mathcal C_{k,1}$. Note both $\mathcal K$ and $\mathcal C_{k,1}$ disjointly form a $\big(\frac {\gamma} 2\big)$-clique. We know that a node $u_m \in \mathcal S_1$, $m=\frac {\gamma+2} 2$, must form a $m$-clique $\mathcal Q$ with either $\mathcal K$ or $\mathcal C_{k,1}$. We swap $u_m \in \mathcal S_1$ with $p \in \mathcal S_2$. Now let $\mathcal P=\mathcal S_2\setminus \mathcal Q$. Note that nodes in $\mathcal P$ form a $(\frac \gamma 2)$-clique such that $\mathcal P\cap \mathcal Q=\emptyset$. Since $\mathcal S_1$ is $\gamma$-reachable even after $u_m$ is swapped out, $\exists u_{m+1} \in \mathcal S_1$ such that $|\mathcal N_{u_{m+1}}\cap \mathcal S_2|\geq \gamma$. There are two cases: 
\begin{enumerate}
    \item $u_{m+1}$ has edges with all $\gamma+1$ nodes in $\mathcal S_2$, or $\gamma$ nodes but not with a node in $\mathcal P$.
    \item $u_{m+1}$ has edges with $\gamma$ nodes in $\mathcal S_2$ but not with a node in $\mathcal Q$.
\end{enumerate}

If 1), $\mathcal Q\cup \{u_{m+1}\}$ forms a $\big(\frac {\gamma+4} 2\big)$-clique. If 2), $\mathcal P\cup\{u_{m+1}\}$ forms a $\big(\frac {\gamma+2} 2\big)$-clique $\mathcal P'$ such that $\mathcal P'\cap \mathcal Q=\emptyset$. In this case, WLOG, let $\mathcal P' \subset \mathcal S_1$ and $\mathcal Q \subset \mathcal S_2$ such that $|\mathcal S_1|=|\mathcal S_2|=\gamma$. Since either $\mathcal S_1$ or $\mathcal S_2$ is $\gamma$-reachable, all of the nodes in either $\mathcal P'$ or $\mathcal Q$ have edges with one additional node in $\mathcal S_2$ or $\mathcal S_1$ respectively, forming a $\big(\frac {\gamma+4} 2\big)$-clique. Thus, $\mathcal G$ must contain a $\big(\frac {\gamma+4} 2\big)$-clique for even $\gamma$. Since $\mathcal G$ must have a clique size of $\frac {\gamma+3} 2$ and $\frac {\gamma+4} 2$ for odd and even $\gamma$ respectively, $\mathcal G$ must have a $\big(\lfloor{\frac {\gamma+4} 2}\rfloor\big)$-clique.
\end{proof}

\section*{References}
\bibliographystyle{IEEEtran}
\bibliography{references_ll}

\begin{IEEEbiography}[{\includegraphics[width=1in,height=1.25in,clip,keepaspectratio]{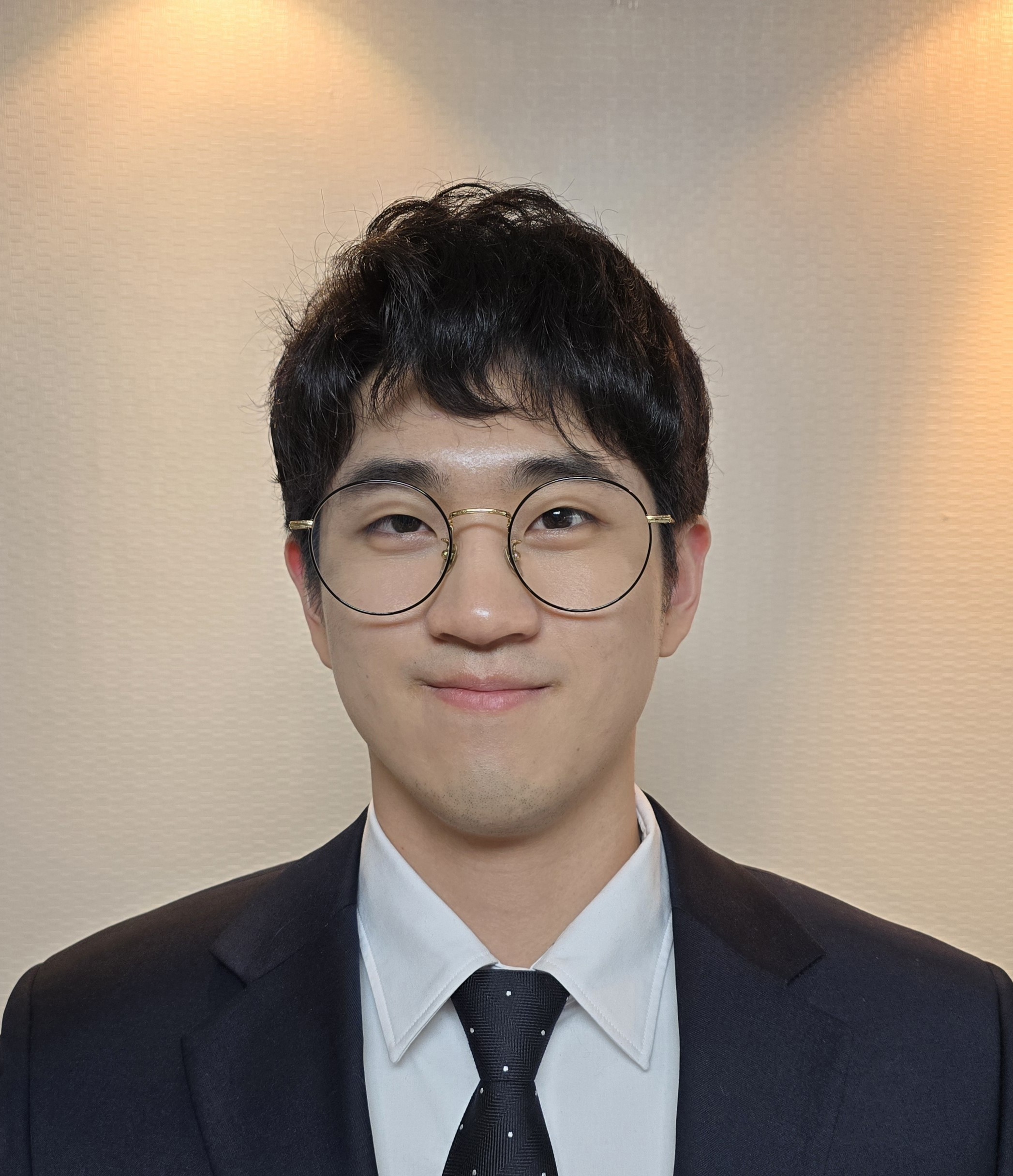}}]{Haejoon Lee} (Student Member, IEEE) received the B.S. degree in applied math and statistics from Stony Brook University, Stony Brook, NY, USA, in 2023. He earned the M.S.
degree in robotics in 2025 from the University of
Michigan, Ann Arbor, MI, USA, where he is currently working toward the Ph.D. degree in robotics, advised by Prof. Dimitra Panagou. 

His research interests include safety, resilience, and security of autonomous systems, with particular emphasis on distributed consensus, optimization, and learning for multi-agent systems in adversarial environments. 

\end{IEEEbiography}

\begin{IEEEbiography}[{\includegraphics[width=1in,height=1.25in,trim={0mm 0 0 0},clip,keepaspectratio]{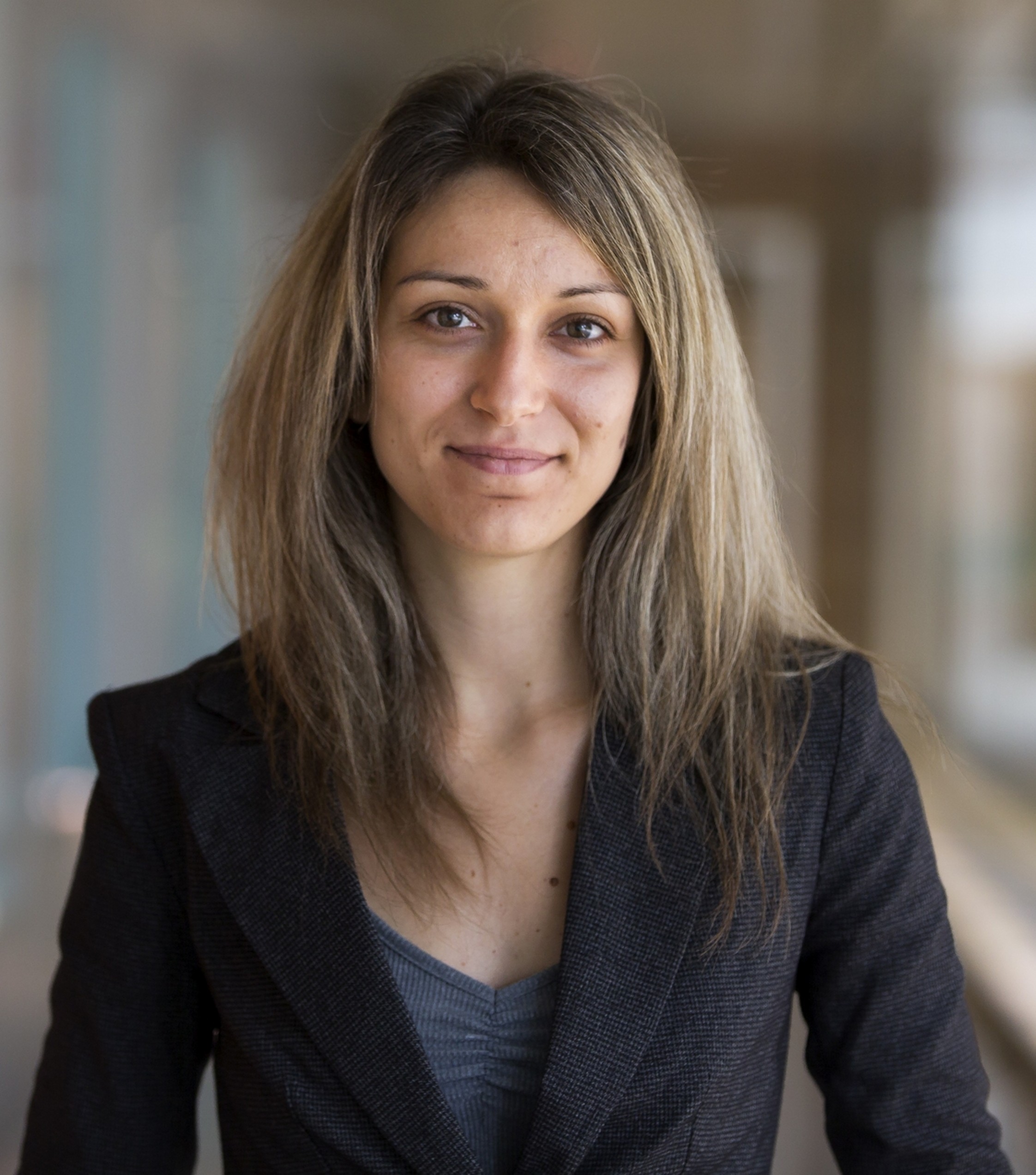}}]{Dimitra Panagou} (Diploma (2006) and PhD (2012) in Mechanical Engineering from the National Technical University of Athens, Greece) is an Associate Professor with the Department of Robotics, with a courtesy appointment with the Department of Aerospace Engineering, University of Michigan. Her research program spans the areas of nonlinear systems and control; multi-agent systems; autonomy; and  aerospace robotics. She is particularly interested in the development of provably-correct methods for the safe and secure (resilient) operation of autonomous systems with applications in robot/sensor networks and multi-vehicle systems under uncertainty. She is a recipient of the NASA Early Career Faculty Award, the AFOSR Young Investigator Award, the NSF CAREER Award, the George J. Huebner, Jr. Research Excellence Award, and a Senior Member of the IEEE and the AIAA.
\end{IEEEbiography}

\end{document}